\documentclass[a4paper,UKenglish,cleveref, autoref, thm-restate]{lipics-v2021}

\hideLIPIcs  

\graphicspath{{./img/}}

\title{Planar Rosa: a family of quasiperiodic substitution discrete plane tilings with $2n$-fold rotational symmetry}

\titlerunning{Planar Rosa}

\author{Jarkko Kari}{University of Turku, Finland}{}{}{}
\author{Victor H. Lutfalla}{ Université Publique, France \\
Université Paris 13, CNRS, LIPN, Villetaneuse, France\\
Aix Marseille Univ, CNRS, LIS, Marseille, France}{victor.lutfalla@lis-lab.fr}{https://orcid.org/0000-0002-1261-0661}{}

\authorrunning{J. Kari and V. H. Lutfalla }

\Copyright{Jarkko Kari and Victor H. Lutfalla} 



\category{} 

\relatedversion{} 

\acknowledgements{We wish to thank Thomas Fernique for his help and proofreading.}

\nolinenumbers 

\EventEditors{John Q. Open and Joan R. Access}
\EventNoEds{2}
\EventLongTitle{42nd Conference on Very Important Topics (CVIT 2016)}
\EventShortTitle{CVIT 2016}
\EventAcronym{CVIT}
\EventYear{2016}
\EventDate{December 24--27, 2016}
\EventLocation{Little Whinging, United Kingdom}
\EventLogo{}
\SeriesVolume{42}
\ArticleNo{23}

\newcommand{\subrosa}{Sub Rosa }
\newcommand{\planarrosa}{Planar Rosa }

\renewcommand{\phi}{\varphi}
\renewcommand{\leq}{\leqslant}
\renewcommand{\geq}{\geqslant}
\newcommand{\R}[1]{\mathbb{R}^{#1}}
\newcommand{\Z}[1]{\mathbb{Z}^{#1}}
\newcommand{\N}{\mathbb{N}}
\newcommand{\tuple}[1]{\left(#1\right)}

\newcommand{\modulo}[1]{\left[ #1 \right]}

\newcommand{\half}{\ensuremath{\frac{1}{2}} }
\newcommand{\thalf}{\ensuremath{\tfrac{1}{2}} }
\newcommand{\nhalf}{\ensuremath{\frac{n}{2}} }
\renewcommand{\epsilon}{\varepsilon}

\newcommand{\trans}{^\mathsf{T}}

\newcommand{\tiling}{\mathcal{T}}
\newcommand{\patch}{\mathbf{P}}

\newcommand{\tileset}{\mathbf{T}}

\newcommand{\vv}[1]{\vec{v}_{#1}}
\newcommand{\ee}[1]{\vec{e}_{#1}}

\newcommand{\reverse}[1]{\overline{#1}}
\newcommand{\abel}[1]{\left[#1\right]}

\newcommand{\boundary}{\partial}

\newcommand{\ie}{\emph{i.e.}, }

\newcommand{\lift}[1]{\widehat{#1}}
\newcommand{\tile}[2]{T_{#1,#2}}
\newcommand{\intsquare}[3]{S_{#1,#2,#3}}
\newcommand{\spanning}[1]{\left\langle #1 \right\rangle}

\renewcommand{\star}[1]{S(#1)}

\newcommand{\slope}{\mathcal{E}}
\newcommand{\elementaryMatrix}[2]{EM_{#1}(#2)}
\newcommand{\expansionMatrix}[1]{M_{#1}}
\newcommand{\elementaryEigenvalue}[3]{\lambda_{#1,#2,#3}}
\newcommand{\identity}[1]{\mathrm{Id}_{#1}}
\newcommand{\eigenMatrix}[1]{N_{#1}}
\newcommand{\hyperplane}[2]{H_{#1,#2}}

\newcommand{\prSubs}[1]{\sigma'_{#1}} 
\newcommand{\prExp}[1]{\phi'_{#1}} 
\newcommand{\prMat}[1]{M'_{#1}} 
\newcommand{\prWord}[1]{\Sigma_{(#1)}} 
\newcommand{\bword}{w }
\newcommand{\prTiling}[1]{\mathcal{T}'_{#1}}
\newcommand{\tgamma}[1]{\ensuremath{\tilde{\gamma}_{#1}}}

\newcommand{\imag}{\mathbf{i}}

\newcommand{\grid}[2]{H(#1,#2)}
\newcommand{\multigrid}[2]{G_{#1}(#2)}
\newcommand{\dualtiling}[2]{P_{#1}(#2)}

\begin{document}

\maketitle

\begin{abstract}
We present Planar Rosa,
a family of rhombus tilings with a $2n$-fold rotational symmetry that are generated by a primitive substitution and that are also discrete plane tilings, meaning that
they are obtained as a projection of a higher dimensional discrete plane. The discrete plane condition is a relaxed version of the cut-and-project condition.
We also prove that the Sub Rosa substitution tilings with $2n$-fold rotational symmetry defined by Kari and Rissanen do not satisfy even the weaker discrete plane condition.
We prove these results for all even $n\geq 4$. This completes our previously published results for odd values of $n$.

\end{abstract}

\section{Introduction}

We study non-periodic rhombus tilings having rotational symmetries. A famous example
by R.~Penrose is a tiling of the plane by copies of two rhombuses that is invariant under a global
five-fold rotation~\cite{penrose1974}. The Penrose tiling has good properties: it is generated by a substitution
process, but it can also be obtained as a projection of a discrete plane in the five dimensional space. More precisely the Penrose tiling satisfies both the \emph{discrete plane} condition and the more restrictive \emph{cut-and-project} condition.
The \emph{discrete plane} condition is that the discrete surface of $\mathbb{R}^5$, from which the tiling is projected, approximates a 2D plane of $\mathbb{R}^5$ by staying within bounded distance of it. This 2D plane is called the \emph{slope} of the tiling.
The \emph{cut-and-project} condition additionally requires that the discrete surface, when projected along the slope onto the orthogonal complement, is a compact set which is the closure of its interior.
The discrete plane condition and the cut-and-project condition guarantee that despite its non-periodicity
the tiling has long-range order, a characteristic feature of
quasicrystals~\cite{baake2013, senechal1996}. We want to construct rhombus tilings with both these good properties and with other orders of rotational symmetry than five. Note that cut-and-project rhombus tilings with any order of rotational symmetry (but with no known substitutive structure) can be constructed by the multigrid method \cite{debruijn1981, gahler1986, lutfalla2021multigrid}.

In~\cite{kari2016}, for any $n\geq 3$, a primitive rhombus substitution called \subrosa was presented that generates a tiling $\tiling_n$ with $2n$-fold global rotational symmetry, see Figures \ref{fig:subrosa4_tiling_banner} and \ref{fig:subrosa6_tiling_banner} for examples of \subrosa tilings.
However, it turns out that these tilings do not have the second desired property: for $n\notin \{3,5\}$,
they are not projections of $n$-dimensional discrete planes, not even in a
relaxed sense where the discrete plane is allowed to deviate within
a uniformly bounded distance from an exact plane of $\mathbb{R}^n$.

\begin{figure}[t]
\centering
\begin{subfigure}{\textwidth}
\includegraphics[width=\textwidth]{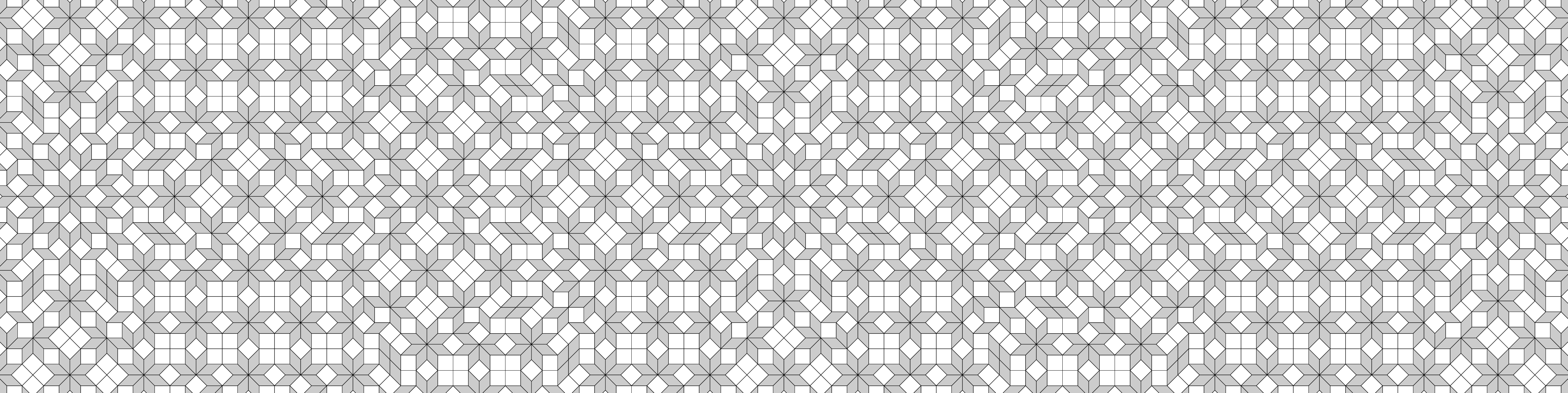}
\caption{$\tiling_4$: the canonical \subrosa $4$ substitution tiling.}
\label{fig:subrosa4_tiling_banner}
\end{subfigure}
\begin{subfigure}{\textwidth}
\includegraphics[width=\textwidth]{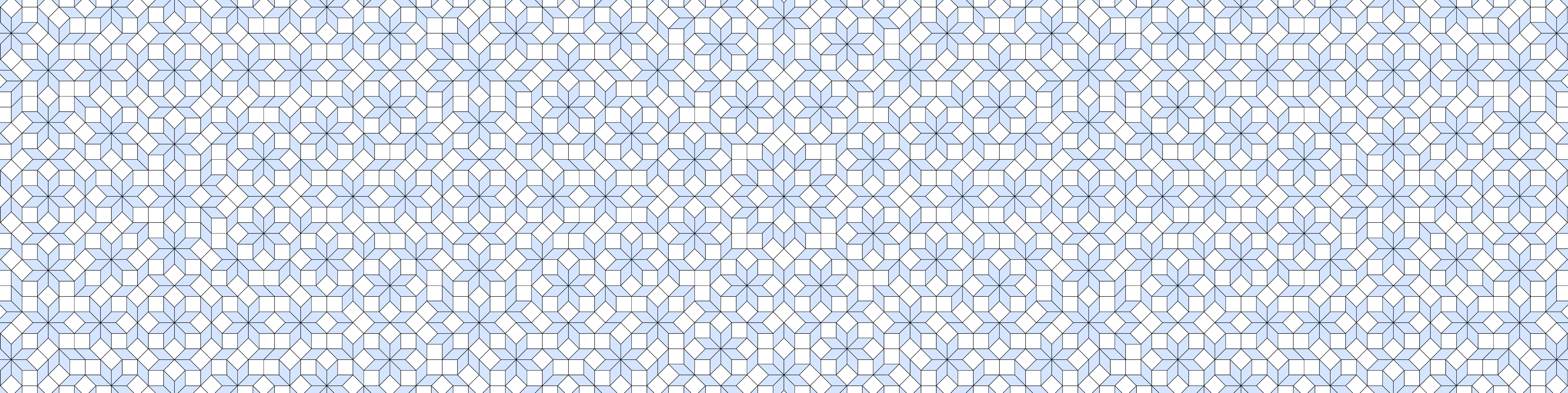}
\caption{$\tiling'_4$: the canonical \planarrosa $4$ substitution discrete plane tiling.}
\label{fig:planarrosa4_tiling_banner}
\end{subfigure}
\caption{Central fragments of canonical \subrosa $4$ and \planarrosa $4$ tilings with $8$-fold rotational symmetry.}
\label{fig:tilings_8fold}
\end{figure}

\begin{theorem}
  \label{th:subrosa_not_planar}
  For any $n\geq 3$ the canonical \subrosa tiling $\tiling_n$ can be lifted to a discrete surface of $\R{n}$ and:
  \begin{enumerate}
  \item The canonical \subrosa tilings $\tiling_3$ and $\tiling_5$ are discrete plane tilings.
  \item For $n=4$ and for any $n\geq 6$ the canonical \subrosa tiling $\tiling_n$ is not a discrete plane tiling.
  \end{enumerate}
\end{theorem}
See Section~\ref{sec:settings} for details of the terminology that is used in the theorem statement.
We provide a proof of Theorem~\ref{th:subrosa_not_planar} for even $n$,
while the case of odd $n$ has been proved in~\cite{kari2021}.

The natural followup question is: for which numbers $n$ does there exist a primitive rhombus
substitution that generates a tiling with $2n$-fold rotational symmetry such that the tiling is
also a projection of a discrete plane in $\mathbb{R}^n$~? For all $n\geq 3$, we present a construction for such a primitive rhombus substitution which we call \planarrosa. See Figures \ref{fig:planarrosa4_tiling_banner} and \ref{fig:planarrosa6_tiling_banner} for examples of \planarrosa tilings.

\begin{theorem}
  \label{th:planar-rosa}
  For any $n\geq 3$ the canonical \planarrosa tiling $\tiling'_n$ is a substitution discrete plane tiling with global $2n$-fold rotational symmetry.
\end{theorem}

\begin{figure}[t]
\centering
\begin{subfigure}{\textwidth}
\includegraphics[width=\textwidth]{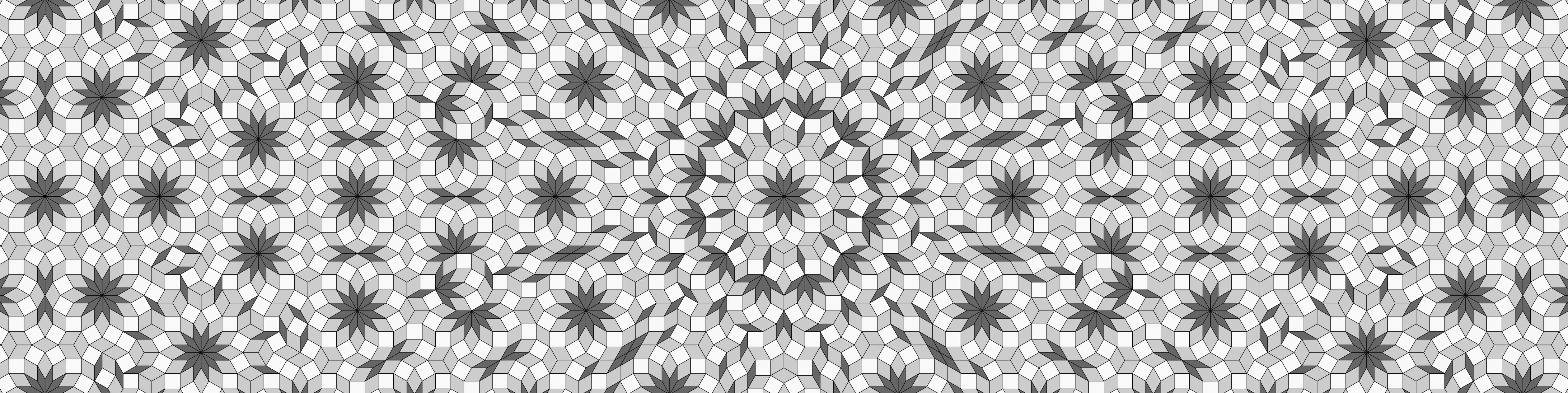}
\caption{$\tiling_6$: the canonical \subrosa $6$ substitution tiling.}
\label{fig:subrosa6_tiling_banner}
\end{subfigure}
\begin{subfigure}{\textwidth}
\includegraphics[width=\textwidth]{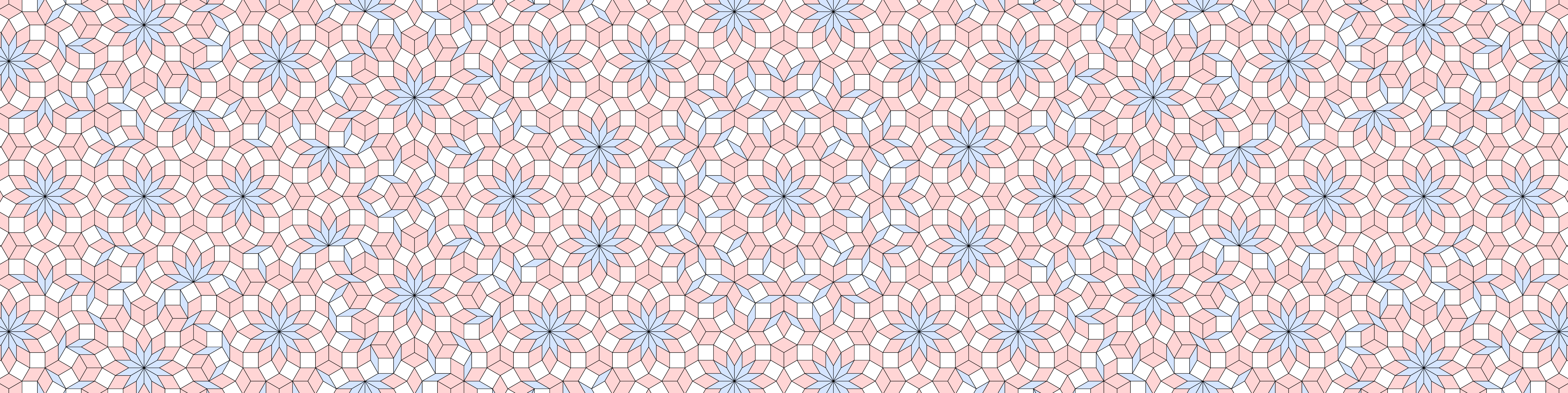}
\caption{$\tiling'_6$: the canonical \planarrosa $6$ substitution discrete plane tiling.}
\label{fig:planarrosa6_tiling_banner}
\end{subfigure}
\caption{Central fragments of canonical \subrosa $6$ and \planarrosa $6$ tilings with $12$-fold rotational symmetry.}
\label{fig:tilings_12fold}
\end{figure}
\noindent
In this paper we give the construction of \planarrosa for all even $n$ and
prove its correctness, see Section \ref{sec:planar-rosa} and Definition \ref{def:planar-rosa}.
The case of odd $n$ was detailed in~\cite{kari2021}. We remark that the tilings
we obtain may not satisfy the cut-and-project condition, the strict condition on the discrete planes
that the Penrose tiling satisfies,
see Section \ref{subsec:discrete_planes} for more details on this.
Our planarity property is weaker, and it remains open whether
our construction provides a cut-and-project tiling, or
whether some other construction exists that provides a cut-and-project solution.

Although the cases of even and odd $n$ in Theorems~\ref{th:subrosa_not_planar} and \ref{th:planar-rosa}
have similarities, there are also significant differences in the proofs, justifying the
present paper.
When $n$ is odd the $n$'th roots of unity (that play a central role in the proofs)
can be lifted into the canonical basis of $\mathbb{R}^n$, while for even $n$
the $n$'th roots of unity come in $\frac{n}{2}$ pairs of opposite numbers and thus
we need to consider $2n$'th roots of unity to get $n$ independent directions. The
\planarrosa substitution is first defined for the edges of tiles, and this edge substitution
is quite different in the odd and even cases of $n$, leading to differences in
the proofs of the primitivity of the substitution and of the tileability of the
interiors of the metatiles.
We remark that also the construction of \subrosa in~\cite{kari2016} was different
in the odd and the even cases of $n$.

The paper is organized as follows. In Section~\ref{sec:settings} we give necessary background and
our notations, including concepts related to rhombus substitutions and discrete planes.
We explain how rhombus tilings and substitutions are lifted into discrete surfaces and surface
substitutions in higher dimensional spaces. We then provide sufficient conditions for the planarity and the non-planarity
of the generated discrete surface, based on the expansion factors of the substitution along its eigenspaces.

In Section~\ref{sec:subrosa} we consider the \subrosa substitutions from~\cite{kari2016}, and prove Theorem~\ref{th:subrosa_not_planar} for even $n$. We recall the construction from~\cite{kari2016},
and using properties of \emph{(pseudo-)circulant} matrices and some trigonometric calculations we show that
the \subrosa substitutions for even $n$ have only expanding eigenspaces and hence the generated tilings
are not discrete planes.

In Section~\ref{sec:planar-rosa} we present the \planarrosa construction for even $n$,
and we prove Theorem~\ref{th:planar-rosa} for even $n$. In all our substitutions
each edge of each rhombus is substituted by the same palindromic sequence of
rhombuses. A suitable distribution of different types of rhombuses on this edge sequence
is the key to guaranteeing the planarity of the generated discrete surface. We
start the section by providing suitable palindromic edge sequences with good
rhombus distributions. The hard part is to prove that with such a rhombus
sequence along the edges, the interior can be tiled with rhombuses as well.
We apply a variant of a tileability criterion by  Kenyon~\cite{kenyon1993},
reformulated using so-called counting functions as in~\cite{kari2021}.

\section{Substitution discrete plane tilings}
\label{sec:settings}
In this section we present the definitions and general ideas for quasiperiodic substitution discrete planes. We start from tilings, then we define quasiperiodicity, substitutions, discrete planes and, finally, lifted substitutions.

\subsection{Rhombus tilings}
A \emph{tiling} is a covering of the Euclidean plane by tiles whose interiors do not overlap. In our case the tiles are unit rhombuses and the tilings are \emph{edge-to-edge} which means that any two distinct tiles that intersect either share a unique common vertex or a full common edge.

The tilings we study are edge-to-edge rhombus tilings with $n$ edge directions for some even $n$. We choose these $n$ directions as $\vv{i}$ for $ i\in \{0,1,2,\dots,n-1\}$ with
\[\vv{i}:= \left(\cos\tfrac{i\pi}{n}, \sin\tfrac{i\pi}{n}\right) = e^{\imag\frac{i\pi}{n}}.\]
Note that we use $\imag$ for the imaginary number, and $i$ for indices.
As we only use the imaginary number in the notations 
$e^{\imag \theta}$ for various $\theta$ there should be no confusion.

We call \emph{prototiles} tiles up to translations. We call \emph{tileset}, denoted by $\tileset$, a set of prototiles. We call $\tileset$-tiling a tiling where all the tiles are translates of the prototiles in $\tileset$.
Here the prototiles we are interested are
\[\tile{i}{j}:= \{ \lambda \vv{i} + \mu\vv{j}, 0\leq \lambda,\mu \leq 1\}\]
for $i,j\in\{0,\dots n-1\}$, $i<j$.
Given a position $x\in\R{2}$ and a prototile $\tile{i}{j}$, we say that $x+\tile{i}{j}$ is the tile of type $i,j$ at position $x$.

\begin{figure}[b]
\includegraphics[width=\textwidth]{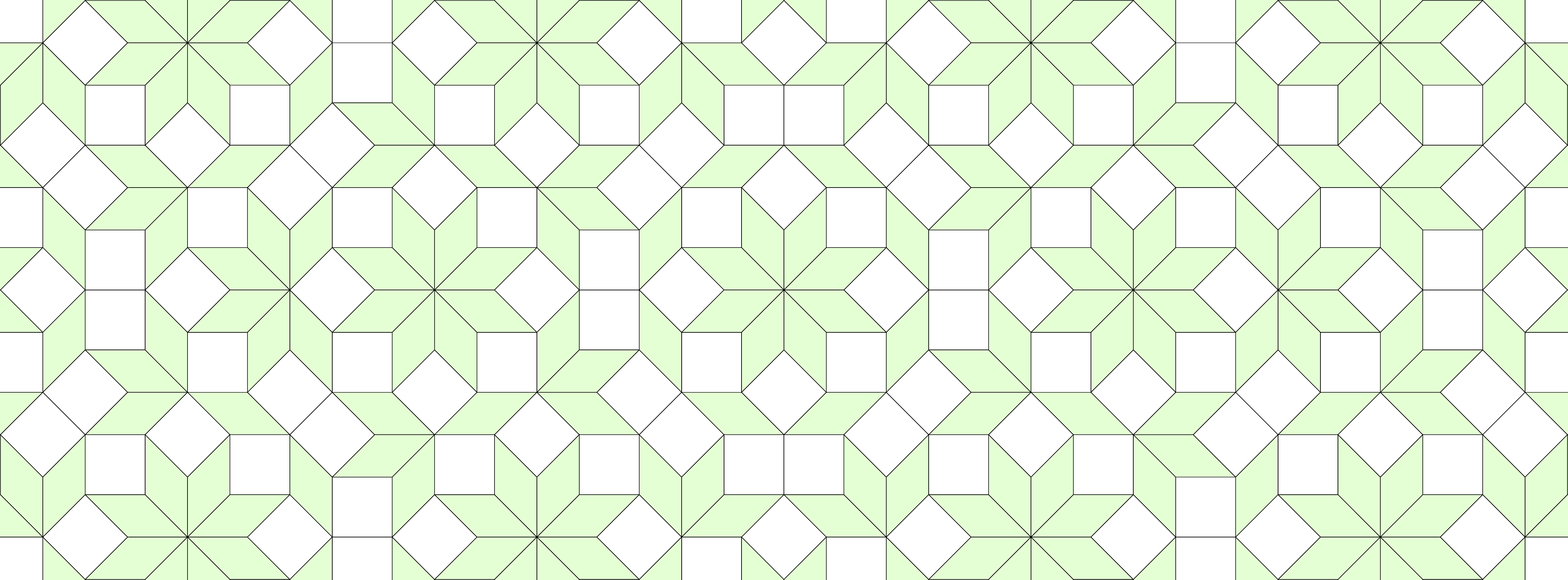}
  \caption{The canonical Ammann-Beenker tiling: an example of edge-to-edge rhombus tiling with $4$ edge directions.}
\end{figure}

We call a \emph{patch} a simply connected set of tiles, and we call a \emph{pattern} a patch up to translations (and possibly rotations). We say that a patch $\patch$ is in the tiling $\tiling$ when $\patch\subset \tiling$, and then we also say that the corresponding pattern appears in the tiling in this position.

A tiling is called \emph{periodic} when there exist a non-trivial translation under which it is invariant. If there is no such translation then it is called \emph{non-periodic}.
A tiling is called \emph{uniformly recurrent} when for every pattern $\patch$ that appears in the tiling, there exists a radius $r$ such that the pattern $\patch$ appears in any ball of radius $r$ of the tiling. A tiling is called \emph{quasiperiodic} \cite{senechal1996} when it is both non-periodic and uniformly recurrent.

A tiling is said to have \emph{global $n$-fold rotational symmetry} when there exist a rotation of angle $\tfrac{2\pi}{n}$ under which the tiling is invariant. A tiling is said to have \emph{local $n$-fold rotational symmetry} when its set of patterns is invariant under rotation of angle $\tfrac{2\pi}{n}$, \ie the image of any pattern of the tiling by a rotation of angle $\tfrac{2\pi}{n}$ is also a pattern of the tiling. For example the famous Ammann-Beenker tiling has local and global 8-fold rotational symmetry \cite{beenker1982},
while the canonical Penrose rhombus tiling has global 5-fold symmetry and local 10-fold symmetry.

\subsection{Substitutions}

A \emph{combinatorial substitution} on a tileset $\mathbf{T}$ is a pair of functions $(\sigma,\partial\sigma)$ where $\sigma$, called the \emph{substitution}, is a function that associates to each prototile $t$  a finite patch of tiles $\sigma(t)$, and $\partial\sigma$, called the \emph{boundary of the substitution}, is a function that to each pair $(t,e)$, where $t$ is a prototile and $e$ is an edge of $t$, associates a set of external edges and/or tiles of $\sigma(t)$.
The substitution $\sigma$ is extended to a function on patches of tiles by applying the substitution separately to each tile and gluing the obtained patches in such a way that it preserves the combinatorial structure, \ie if a patch P is exactly two tiles $t_0$ and $t_1$ adjacent along an edge $e$ then $\partial\sigma(t_0,e)$ and $\partial\sigma(t_1,e)$ have to be equal and the patch $\sigma(P)$ is obtained from gluing the two patches $\sigma(t_0)$ and $\sigma(t_1)$ along the set of edges and/or tiles $\partial\sigma(t_0,e)=\partial\sigma(t_1,e)$ (see Figure \ref{fig:combinatorialsubstitution}), note that the absence of compatibility issues in this gluing process is a strong condition.

\begin{figure}[!b]
  \includegraphics[width=\textwidth]{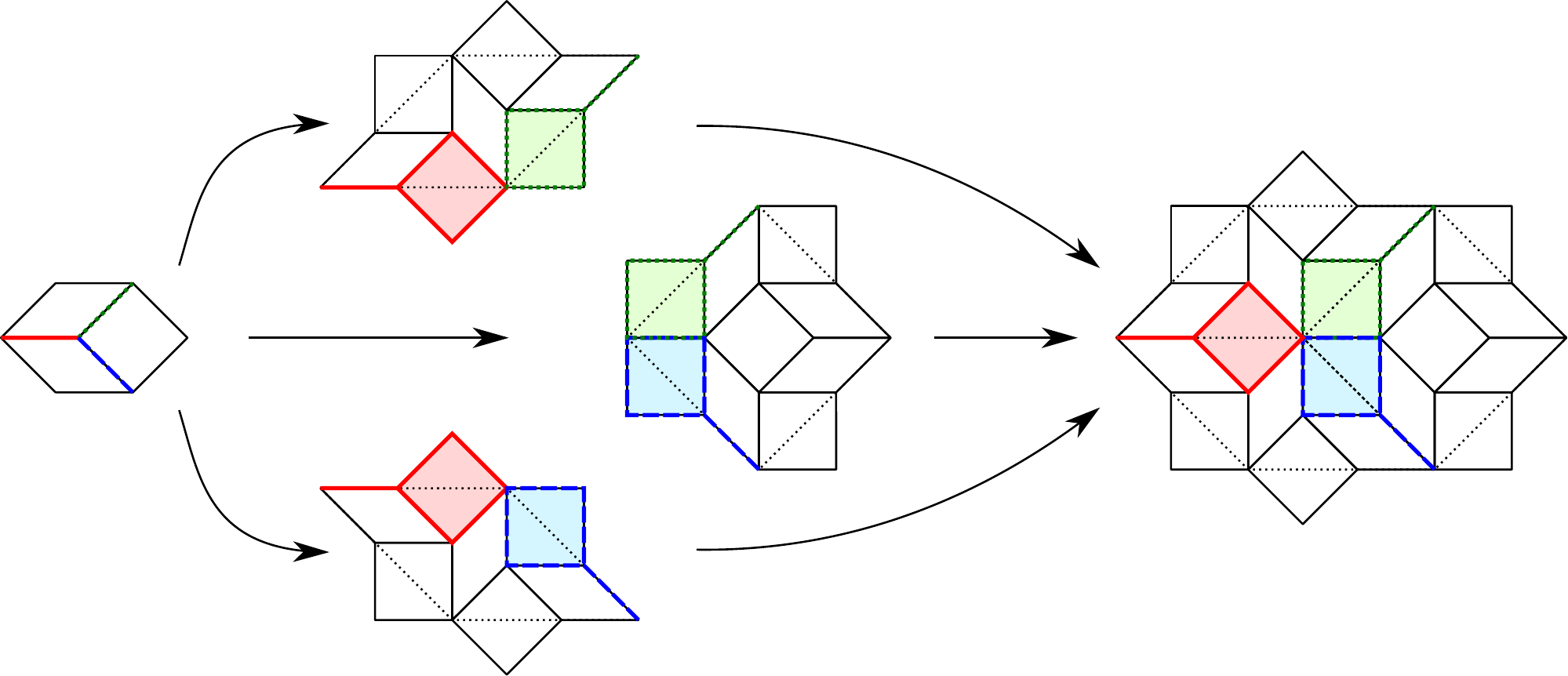}
  \caption{An example of a combinatorial substitution: on the left the initial patch with the three internal edges coloured, in the middle the images of the three initial tiles with the image of the internal edges coloured also, and on the right the patch obtained by gluing the image of the three tiles along the image of the shared edges.}
  \label{fig:combinatorialsubstitution}
\end{figure}

We call \emph{metatiles of order $k$} of $\sigma$ the patterns $\sigma^k(t)$ for $t\in\mathbf{T}$. We usually say \emph{metatiles} for the first order metatiles and we call \emph{sequence of metatiles} the sequence $(\{\sigma^k(t), t \in \tileset\})_{k \in \mathbb{N}}$

We say that a combinatorial substitution $(\sigma,\partial\sigma)$ is \emph{well-defined} when the substitution $\sigma$ can be extended to its whole sequence of metatiles, \ie there is no metatile $\sigma^k(t)$ on which $(\sigma,\partial\sigma)$ cannot be applied in a way that respects the combinatorial structure. In the following we always assume that the substitutions are well-defined and we usually omit the boundary $\partial\sigma$.

A finite pattern is called \emph{legal} for substitution $\sigma$ if it appears in some $\sigma^k(t)$ with $t\in\tileset$ and $k\in\mathbb{N}$. A tiling $\tiling$ is called \emph{legal} for $\sigma$ if every finite pattern of $\tiling$ is legal for $\sigma$. We also then say that $\tiling$ is a $\sigma$-tiling.

A \emph{vertex-hierarchic substitution} is a combinatorial substitution $(\sigma,\partial\sigma)$ such that there exists an expansion $\phi$ (orientation preserving expanding similitude of the plane) such that for any tile $t$ and any edge $e$ of $t$,
the vertices of the expanded edge $\phi(e)$ are vertices of the boundary $\boundary(t,e)$
and
the area spanned by $\sigma(t)$ is equal to the area of the expanded tile $\phi(t)$ where boundary tiles (tiles that are in some $\partial\sigma(t,e)$) count only for half in the computation of the area, \ie
\[ Area\left(\phi(t)\right) = Area\left(\bigcup\limits_{t' \in \mathring{\sigma}(t)} t'\right) + \frac{1}{2}Area\left(\bigcup\limits_{t' \in \partial\sigma(t)} t' \right) \qquad V\left(\phi(e)\right)\subset V\left(\bigcup\limits_{x \in \partial\sigma(t,e)} x\right). \]
See for example the Ammann-Beenker substitution in Figures \ref{fig:combinatorialsubstitution} and \ref{fig:ammann-beenker_subst} where 
the expanded tiles are represented in dotted lines.
\begin{figure}[t]
\includegraphics[width=\textwidth]{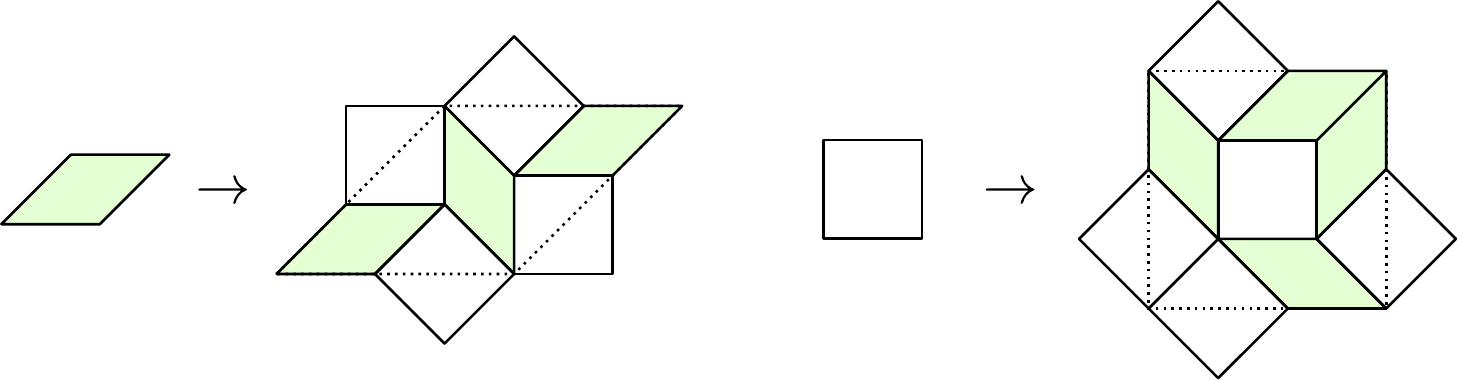}
  \caption{The Ammann-Beenker substitution, an example of vertex-hierarchic substitution on two rhombus tiles up to translation and rotation.}
  \label{fig:ammann-beenker_subst}
\end{figure}

The idea of this family of substitutions is that it is an inflation-subdivision process \cite{grunbaum1987, baake2013} meaning that the substitution first inflates a prototile $t$ to $\phi(t)$ (which is not a tile) and then subdivides $\phi(t)$ to obtain a patch of tiles. The subdivision can differ from $\phi(t)$ but the vertices of $\phi(t)$ must be boundary vertices of $\sigma(t)$ and the area of $\phi(t)$ and $\sigma(t)$ must be equal.

Note that substitution tilings are also sometimes called inflation tilings, or self-similar tilings, though often only \emph{edge-hierarchic} substitutions are considered in these cases \cite{baake2013,solomyak1998}, \ie substitutions where, for any tile $t$, the union of the tiles in $\sigma(t)$ is exactly the expanded tile $\phi(t)$.

\begin{figure}[b]
\includegraphics[width=\textwidth]{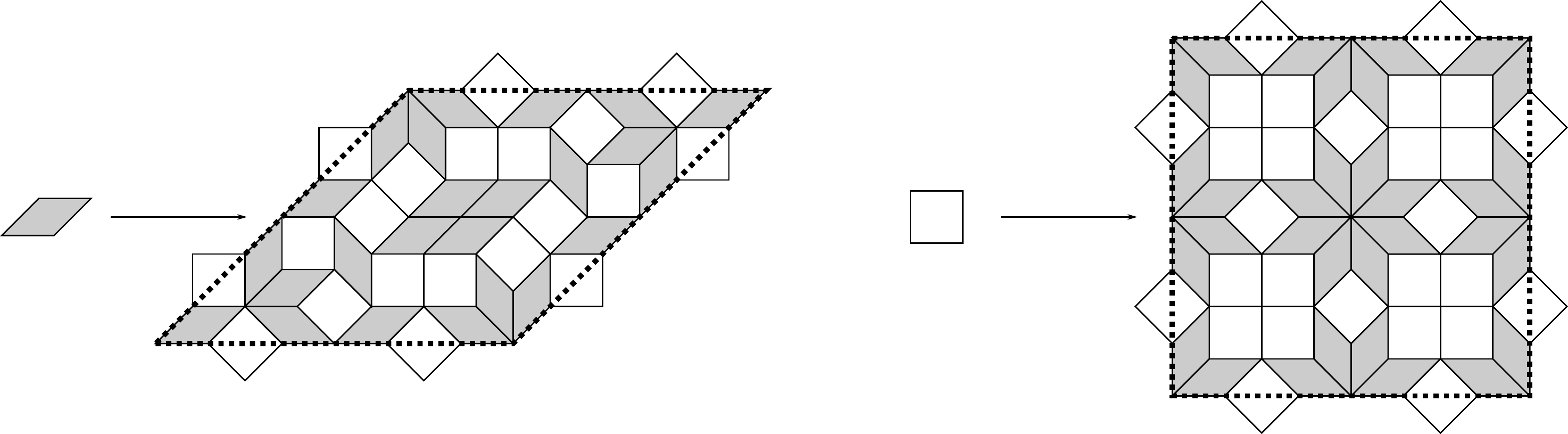}
  \caption{The Sub Rosa 4 substitution, an example of vertex-hierarchic substitution where the image of all edges is the same up to rotation and translation.}
  \label{fig:subst_subrosa4}
\end{figure}

In this paper we study vertex-hierarchic substitutions with the additional property that
the images $\partial\sigma(t,e)$ of all edges $e$ of all tiles $t$ are the same up to rotation and translation.
For an example, see the Sub Rosa 4 substitution in Figure \ref{fig:subst_subrosa4}.
On the other hand, the well-known Penrose substitution and Ammann-Beenker substitution are not in this class.
This restriction reduces the possibilities for the shape of the boundary, but it simplifies conditions on the substitution, ensures that the substitutions are well-defined and makes it easy to lift the substitution in $\R{n}$. In this case we call \emph{edgeword} the sequence of rhombuses and/or edges in the image of an edge up to translation and rotation. Note that we call this sequence of rhombuses and/or edges a word because we consider it as a sequence of symbols/letters rather than a sequence of geometrical shapes.

A substitution is called \emph{primitive of order} $k$ with $k\in \mathbb{N}$ when for any prototile $t$, the metatile $\sigma^k(t)$ of order $k$ 
contains all the different prototiles of the tileset. The substitution is \emph{primitive} if it is primitive of order $k$ for some $k$.
\begin{lemma}
  If a tiling $\tiling$ is legal for some primitive substitution $\sigma$ then $\tiling$ is uniformly recurrent.
\end{lemma}
This result is well-known and can be found in \cite[\S 5]{baake2013}.

Given a substitution $\sigma$ we call \emph{seed} (for $\sigma$) a patch $S$ of tiles such that $S$ is at the centre of $\sigma(S)$ and such that the circle inscribed in $\sigma(S)$ strictly contains the circle circumscribed to $S$. A seed $S$ defines a sequence of increasing (for inclusion) patches $(\sigma^k(S))_{k\in \mathbb{N}}$ and a limit tiling $\tiling:= \lim\limits_{k\to\infty} \sigma^k(S)$, this limit tiling is also called fixpoint of $\sigma$ from seed $S$. A seed $S$ is called legal if it is a legal patch for $\sigma$, \ie if it appears in $\sigma^k(t)$ for some tile $t$ and integer $k$.

\begin{lemma}
  Let $\sigma$ be a vertex-hierarchic substitution. Let $S$ be a legal seed for $\sigma$.
  The limit tiling $\tiling:= \lim\limits_{k\to\infty} \sigma^k(S)$ is legal for $\sigma$. 
  Furthermore if $S$ has global $n$-fold rotational symmetry around its centre vertex then so does $\tiling$.
\end{lemma}

\begin{proof}
  This lemma is quite straightforward but let us give the main elements.

  The limit tiling exists because it is the limit of an increasing sequence (for inclusion) of patches, and the limit is a full tiling of the plane because the sequence of circles inscribed in $\sigma^k(S)$ covers $\R{2}$ due to the condition that $\sigma$ expands $S$ in all directions.

  The limit tiling is legal for $\sigma$. Indeed, $S$ is legal so it is in some $\sigma^{k_0}(t_0)$ with $k_0\geq 0$ and $t_0\in\tileset$, and any patch of $\tiling$ is included in some $\sigma^{k_1}(S)\subset \sigma^{k_0+k_1}(t_0)$.

  If $S$ has global $n$-fold rotational symmetry, then so does $\sigma^k(S)$ for any $k$, and therefore so does $\tiling$.
\end{proof}

\subsection{Discrete planes}
\label{subsec:discrete_planes}
For any positive integer $n$, let $(\ee{i})_{0\leq i < n}$ 
be the canonical basis of $\R{n}$.
For $i<j$, we define $\intsquare{x}{i}{j}$, the integer square with position $x\in\Z{n}$ and type $\{i,j\}$, by
\[ \intsquare{x}{i}{j}:= x + \{ \lambda \ee{i} + \mu \ee{j}, 0\leq \lambda,\mu\leq 1\}.\]

Any edge-to-edge tiling with unit rhombus tiles and $n$ edge directions can be lifted to a discrete surface of $\R{n}$ by choosing an origin vertex which is sent to the origin of $\R{n}$ and mapping each edge direction to a vector of the canonical basis \cite{levitov1988}. In this lifting operation unit rhombus tiles are mapped to unit integer squares of $\R{n}$.

We call such a tiling a \emph{discrete plane tiling} when the lifted discrete surface stays within bounded distance of some two-dimensional plane $\slope$, which is called the \emph{slope} of the tiling.
Note that \emph{discrete plane tilings} are sometimes called \emph{planar tilings} \cite{bedaride2015}, and that in the scope of tilings with local matching conditions they are called \emph{weak local rules tilings} \cite{levitov1988} or \emph{weak matching rules tiling} \cite{socolar1990}.

The reason why any edge-to-edge rhombus tiling with $n$ edge directions can be lifted is that in an edge-to-edge rhombus tiling any two edge paths from the origin to a vertex $x$ are identical up to reordering and cancellation of edges (\ie $\vv{i} -\vv{i} = 0$).
In particular, there is a unique abelianized path from the origin to $x$~\cite{levitov1988}. The lifted vertex $\lift{x}$ is characterized uniquely by
\[  \lift{x} = (k_0,\dots k_{n-1}) \in \Z{n} \Leftrightarrow x=\sum\limits_{0\leq i < n} k_i \vv{i}, \]
where $\sum\limits k_i \vv{i}$ is the abelianized edge path from 0 to $x$ in the tiling.

Note that the lift of the tile $\tile{i}{j}$ at position $x$ is the square $\intsquare{\lift{x}}{i}{j}$, \ie
\[\lift{x+\tile{i}{j}}:= \intsquare{\lift{x}}{i}{j}.\]

We mostly use the same names and symbols for the objects in the Euclidean plane and their lifted counterparts in $\R{n}$. However, when we want to emphasize the difference we denote $\lift{x}$ for the lifted version of an object $x$.

\subsection{Lifted substitutions}
\label{subsec:lifted_substitutions}
In the previous section we presented tilings, substitutions and discrete planes. Since we study substitution discrete plane tilings, or more generally substitution tilings lifted in $\R{n}$, we now present how substitutions and expansions are lifted in $\R{n}$.

Let us take a vertex-hierarchic substitution $\sigma$ and its expansion $\phi$ such that the images by the substitution of any two parallel edges $e$ and $e'$ are the same up to a translation, \ie for any two tiles $t$ and $t'$, and for any edge $e$ of $t$ and any edge $e'$ of $t'$ if $e$ and $e'$ are parallel then $\partial\sigma(t,e)$ and $\partial\sigma(t',e')$ are equal up to translation.
\begin{figure}[b]
    \center
    \includegraphics[width=0.6\textwidth]{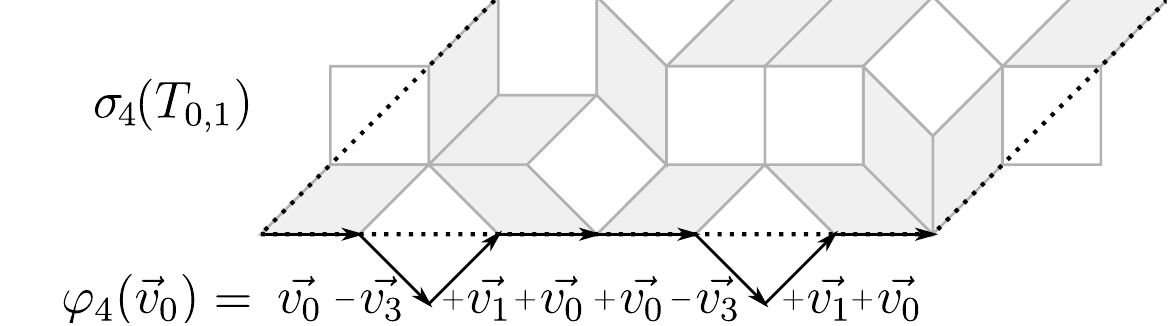}
    \caption{$\phi_4(\vv{0})$ is the abelianized path from the origin to the $\vv{0}$ corner in the metatile $\sigma_4(\tile{0}{1})$, \ie $\phi_4(\vv{0}) = 4\vv{0} + 2\vv{1} - 2\vv{3}$.}
    \label{fig:phi_vi}
\end{figure}

We define the lifted substitution $\lift{\sigma}$ from discrete surfaces (of $\R{n}$) to discrete surfaces (of $\R{n}$), and the lifted expansion $\hat{\phi}$ from $\R{n}$ to $\R{n}$ as follows.
\begin{itemize}
  \item we define the lifted substitution $\hat{\sigma}$ on the prototiles as
    \[\lift{\sigma}(\hat{r}):=\lift{\sigma(r)},\]
    for any rhombus prototile $r$. Note that, by definition, the prototiles contain a vertex on the origin, since the substitution $\sigma$ is vertex-hierarchic this origin vertex has an image in $\sigma(r)$ which is used as the origin to lift $\sigma(r)$. Since the prototiles are the tiles $\tile{j}{k}$, this defines $\lift{\sigma}$ for the integer squares $\intsquare{0}{j}{k}$.
  \item we define the lifted expansion $\lift{\phi}$ so that we can define the lifted substitution on all integer squares, at all positions.
    We define $\lift{\phi}$ as a linear map of $\R{n}$ by its image on the canonical basis of $\R{n}$.
    For $i\in{0,..,n-1}$ we define $\lift{\phi}(\ee{i})$ as
    \[ \lift{\phi}(\ee{i}):= \lift{\phi(\vv{i})},\]
    where $\lift{\phi(\vv{i})}$ is the lifted version of the abelianized path from the origin to the $\vv{i}$ corner vertex in the patch $\sigma(\tile{i}{j})$ for some $j\neq i$, see Figure \ref{fig:phi_vi}. Note that $\phi(\vv{i})$ is uniquely defined due to the condition that the image by the substitution of any two parallel edges is the same up to a translation.
  \item finally, we define the lifted substitution $\lift{\sigma}$ on any integer square as
    \[ \lift{\sigma}(\intsquare{x}{i}{j}):= \lift{\phi}(x) + \lift{\sigma}(\intsquare{0}{i}{j}) = \lift{\phi}(x) + \lift{\sigma(\tile{i}{j})},\]
    and extend it on any discrete surface by linearity.
\end{itemize}

Note that we quite often consider the expansion of a lifted tile that is a unit integer square, \ie a unit square of $\mathbb{R}^n$ with vertices in $\mathbb{Z}^n$. The image of the lifted tile by the expansion is an integer rhombus, \ie a rhombus of $\mathbb{R}^n$ with vertices in $\mathbb{Z}^n$.
Remark that the image by $\lift{\sigma}$ of any integer square $\intsquare{0}{i}{j}$ has the same shape as $\lift{\phi}(\intsquare{0}{i}{j})$ in the sense that the vertices of $\lift{\phi}(\intsquare{0}{i}{j})$ are extremal boundary vertices of $\lift{\sigma}(\intsquare{0}{i}{j})$.

Since $\sigma$ and $\lift{\sigma}$ have exactly the same behaviour, we write $\sigma$ for both, and the context makes it clear if we are considering the $\R{n}$ versions or the $\R{2}$ versions.

To study the behaviour of the linear map $\phi$ in $\R{n}$ we consider its matrix,
which we denote by $\expansionMatrix{\phi}$ or simply $\expansionMatrix{}$. The matrix operates on column vectors from the right, \ie $\phi(\vec{r})\trans = \expansionMatrix{\phi}\cdot \vec{r}\trans$, where $\vec{r}$ is a line vector and $\vec{r}\trans$ a column vector.
We use the notation $\phi$ when we apply the expansion as a function, and the notation $\expansionMatrix{\phi}$ or $\expansionMatrix{}$ when we study it as a matrix.

\begin{proposition}[Sufficient conditions for planarity and for non-planarity \cite{kari2021}]
\label{prop:planarity_non-planarity}
Let $\sigma$ be a primitive vertex-hierarchic substitution such that the image of any two parallel edges is the same up to translation, and let $\phi$ be its expansion.
\begin{enumerate}
\item If there exists a plane $\slope$ of $\mathbb{R}^n$ such that $\phi$ is expanding along plane $\slope$ and strictly contracting along its orthogonal complement, then any tiling $\tiling$ that is legal for $\sigma$ is a discrete plane of slope $\slope$.
\item If there exists a subspace $\mathcal{V}$ of $\mathbb{R}^n$ of dimension at least 3 such that $\phi$ is expanding along $\mathcal{V}$, then tilings admissible for $\sigma$ are not discrete planes.
\end{enumerate}
\end{proposition}
For a full proof see \cite{kari2021} (or \cite[\S 3.3]{lutfalla2021these} for more details and figures). A similar proof of the first item and all the elements for the proof of the second item can be found in \cite{arnoux2011, arnoux2001higher}. For completeness we give a quick proof sketch.
\begin{proof}[Proof sketch.]
The idea of the proof is quite simple.
In the first case, the metatiles of high-order of $\sigma$ are very flat along $\slope$ because the expansion $\sigma$ tends to expand along $\slope$ and contract along $\slope^\bot$. As any finite patch of a $\sigma$-tiling appears in a metatile, it means that the full tiling is also very flat along $\slope$.

In the second case, since the set of edges of the prototiles lift to the canonical basis of $\mathbb{R}^n$ we know that for any decomposition $\mathbb{R}^n = \mathcal{V}\oplus \mathcal{W}$ the projection of the edges of prototiles onto $\mathcal{V}$ parallel to $\mathcal{W}$ form a generating family of $\mathcal{V}$, so the metatiles of the substitution are non-trivial and bigger and bigger along $\mathcal{V}$. So the metatiles are not flat or planar along any slope and nor are the $\sigma$-tilings.
\end{proof}
In this work we are interested only in these two cases, the case in between where the expansion might have a eigenvalue $1$ on some subspace is not considered here.

\section{Sub Rosa substitution tilings}
\label{sec:subrosa}
In this section we briefly present the construction for the Sub Rosa substitution tilings  for even $n$, as defined in \cite{kari2016}. We then present how to lift the Sub Rosa substitutions in $\R{n}$ for even $n$, and we compute the eigenvalues of the Sub Rosa expansions to prove Theorem \ref{th:subrosa_not_planar} for even values of $n$. Recall that the proof of Theorem \ref{th:subrosa_not_planar} for odd $n$ is presented in \cite{kari2021}.

\subsection{Construction}
The Sub Rosa tilings form a family of substitution rhombus tilings with a global $2n$-fold rotational symmetry \cite{kari2016}. We only consider the case for even $n$ since the two constructions for odd and even $n$ are somewhat different.
In the Sub Rosa construction the substitution rule is given by the edgeword of the substitution, \ie the sequence of edges and rhombuses along the edge of the substitution.
The interior is then tiled using a variant of the Kenyon criterion for the tileability of a polygon with parallelograms \cite{kenyon1993}.

The edgeword $\Sigma(n)$, for even $n$, is given by
\[\Sigma(n):= s(n)\cdot \reverse{s(2)} \cdot \reverse{s(4)} \cdot \dots \cdot \reverse{s(n-4)} \cdot \reverse{s(n-2)}\ |\
s(n-2) \cdot s(n-4) \cdot \dots \cdot s(4) \cdot s(2) \cdot \reverse{s(n)},\]
where $s(k)$ is the word of even numbers from $0$ to $k-2$, \ie $s(k):=024\dots (k-2)$, $\reverse{w}$ is the mirror image of word $w$, $\cdot$ is the word concatenation, the symbol $|$ represents the middle of the word,  the letter $0$ encodes a unit edge and the even non-null integer $k$ encodes a rhombus of angle $\tfrac{k\pi}{n}$ (bisected through its angle $\tfrac{k\pi}{n}$).
For example the edgeword $\Sigma(4)=020020$ means that along the edges of the substitution $\sigma_4$  we find a unit edge, a bisected rhombus of angle $\tfrac{\pi}{2}$, two unit edges, a bisected rhombus of angle $\tfrac{\pi}{2}$, and a unit edge (see Figure \ref{fig:subst_subrosa4}).

\begin{table}
\begin{center}
    \begin{tabular}{l c}
      $\Sigma(0)$ & | \\
    $\Sigma(2)$ & 0|0 \\
    $\Sigma(4)$ & 020|020\\
    $\Sigma(6)$ & 024020|020420\\
    $\Sigma(8)$ & 0246020420|0240206420 \\
    $\Sigma(10)$ & 024680204206420|024602402086420 \\
  \end{tabular}
  \end{center}
\caption{Table of $\Sigma(n)$ for small $n$.}
\label{table:sigma_n_even}
\end{table}
\begin{figure}[b]
\includegraphics[width=\textwidth]{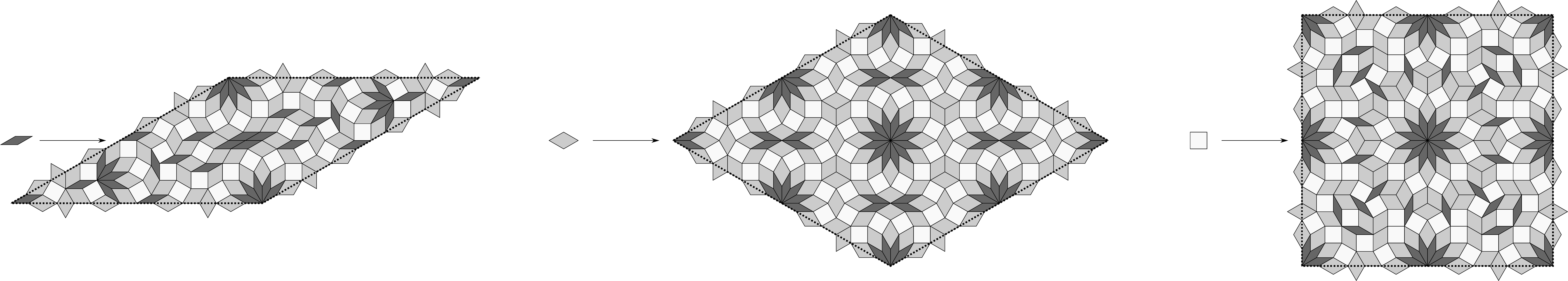}
\caption{The Sub Rosa 6 substitution.}
\end{figure}

\subsection{Lifting to $\R{n}$}
\label{subsec:srLifting}
In the following we consider a vertex-hierarchic substitution $\sigma$ on rhombus tiles with the $2n$'th roots of unity
\[ \vv{i}:= \left(\cos\tfrac{i\pi}{n}, \sin\tfrac{i\pi}{n}\right) = e^{\imag \frac{i\pi}{n}}, \text{ for } i \in \{0,1,\dots,n-1\}.\]
Note that we use the first $n$ of the
$2n$'th roots of unity because $n$ is even and we want $n$ different directions.
Our choice of $\vv{i}$ means the in the positive rotation ordering in the plane we have
\[ \vv{0} \prec \vv{1} \prec \vv{2} \dots \prec \vv{n-1} \prec -\vv{0} \prec -\vv{1} \dots \prec -\vv{n-1} \prec \vv{0},\] see for example Figure \ref{fig:edge-directions_6} for the vectors $\vv{i}$ in the Sub Rosa 6 tilings.
\begin{figure}
\center
\includegraphics[width=0.8\textwidth]{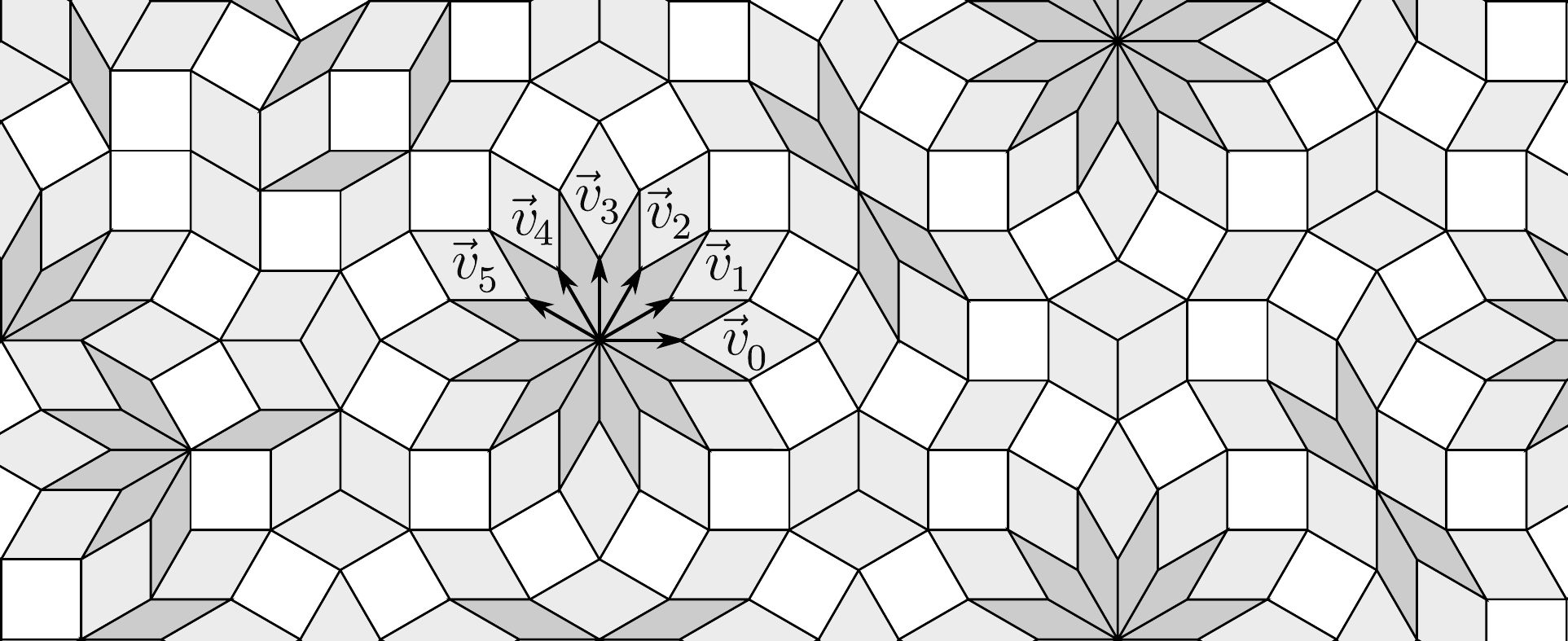}
\caption{Edge directions in Sub Rosa 6.}
\label{fig:edge-directions_6}
\end{figure}

We then lift the rhombus tilings to $\R{n}$, as described in Subsection \ref{subsec:discrete_planes}. We decompose $\R{n}$ in $\tfrac{n}{2}$ planes:
We define the plane $\slope_n^k$ for $0\leq k <\tfrac{n}{2}$ by its two generating vectors
\[\slope_n^k:= \spanning{ \left(\cos\tfrac{(2k+1)i\pi}{n}\right)_{0\leq i < n}, \left(\sin\tfrac{(2k+1)i\pi}{n}\right)_{0\leq i < n}}_{\R{}} \]
which we also write as a single complex generating vector for simplicity
\[ \slope_n^k:= \spanning{\left(e^{\imag\frac{(2k+1)i\pi}{n}}\right)_{0\leq i < n}}.\]
Note that the planes $\slope_n^k$ are pairwise orthogonal and that $\R{n}$ is their direct sum. For more details on this, see Lemma \ref{lemma:decomposition_Rn} in the Appendices.

We lift the Sub Rosa substitution $\sigma_n$ and its associated expansion $\phi_n$ in $\R{n}$ as described in Subsection \ref{subsec:lifted_substitutions}. Note that for the Sub Rosa substitutions we also assume that the image of any edge by the substitution is the same up to rotation and translation.
Given the way we lift $\vv{i}=e^{\imag \frac{i\pi}{n}}$ to $\ee{i}$, this means that the lifted expansion $\phi_n$ and its associated matrix $\expansionMatrix{}$ are \emph{pseudo-circulant}, as defined below.

We call a matrix $(m_{i,j})_{0\leq i,j < n}$   \emph{pseudo-circulant} if
 $m_{i,j} = m_{i',j'}$ when $i-j = i'-j'$ and $m_{i,j} = - m_{i',j'}$ when $i-j = i'-j' \pm n$. For example, see the matrices $\expansionMatrix{4}$ and $\expansionMatrix{6}$ below (expansion matrices of the Sub Rosa expansions $\phi_4$ and $\phi_6$).

\begin{align} \expansionMatrix{4} = \begin{pmatrix} 4 & 2 & 0 & - 2 \\ 2 & 4 & 2 & 0 \\ 0 & 2 & 4 & 2 \\ -2 & 0 & 2 & 4\end{pmatrix} \qquad \qquad
  \expansionMatrix{6} = \begin{pmatrix} 6 & 4 & 2 & 0 & -2 & -4 \\ 4 & 6 & 4 & 2 & 0 & -2 \\ 2 & 4 & 6 & 4 & 2 & 0 \\ 0 & 2 & 4 & 6 & 4 & 2 \\ -2 & 0 & 2 & 4 & 6 & 4\\ -4 & -2 & 0 & 2 & 4 & 6 \end{pmatrix}.
  \label{eq:m4_m6}
  \end{align}
The fact that the Sub Rosa expansion $\phi_n$ is a pseudo-circulant linear map is simply due to the way we order the set of edge vectors $\{\pm \vv{k}, 0 \leq k <n\}$ together with the definition of $\sigma_n$ by its edgeword $\Sigma(n)$, \ie a unique sequence of edges and rhombuses which appears on each edge of each metatile. The definition of $\sigma_n$ by its edgeword implies that the image of $\vv{0}$ and $\vv{1}$ are the same up to a permutation (and sign inversion) of the edge vectors $(\vv{k})_{0\leq k < n}$, this permutation is determined by the way we order the set of edge vectors, since in the positive rotation ordering we have $\vv{0} \prec \vv{1} \dots \prec \vv{n} \prec -\vv{0} \dots$ we indeed obtain pseudo-circulant linear maps. See Figures \ref{fig:edge-directions_6} and \ref{fig:lifting_sigma6}.

\begin{figure}[b]
\center
\includegraphics[width=0.8\textwidth]{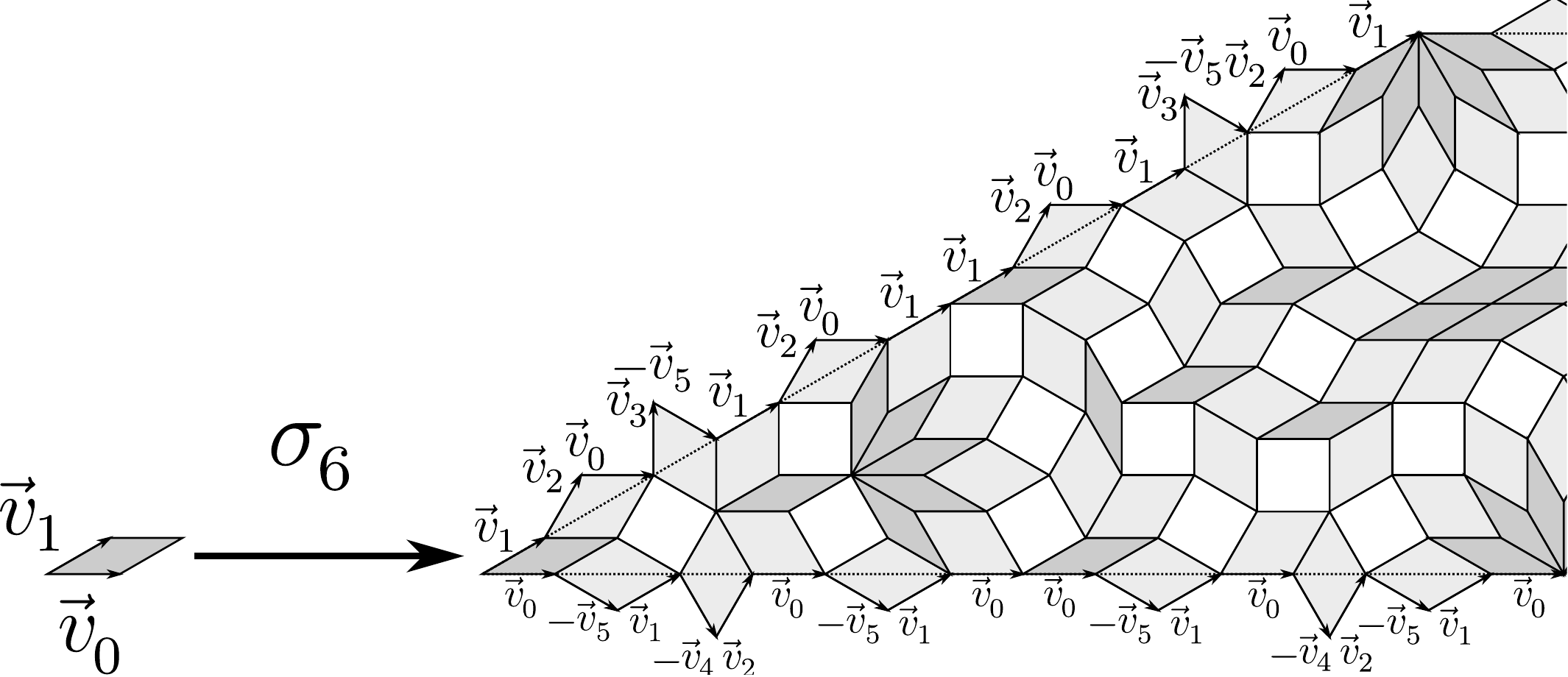}
    \caption{Lifting $\sigma_6$.}
    \label{fig:lifting_sigma6}
\end{figure}

Pseudo-circulant matrices are quite similar to circulant matrices, \ie constant diagonal matrices with
$m_{i,j} = m_{i',j'}$ whenever $i-j \equiv i'-j' \modulo{n}$.
The eigenvalues and eigenvectors of circulant matrices are well-known~\cite{Davis1979}, and this directly provides
the following result on pseudo-circulant matrices.
\begin{proposition}
\label{prop:preudocirculant}
Let $\zeta\in\mathbb{C}$ be such that $\zeta^n=-1$, \emph{i.e.}, a $2n$'th root of unity that is not an $n$'th root of unity.
Let $M$ be an $n\times n$ pseudo-circulant matrix whose first row is $(m_0,m_1,\dots, m_{n-1})$. Then $(\zeta^k)_{0\leq k <n}$ is an eigenvector of $M$ with eigenvalue $\sum_{j=0}^{n-1} m_j\zeta^j$, \ie
\[
M\cdot (1,\zeta,\dots,\zeta^{n-1})\trans
=  \left(\sum_{j=0}^{n-1} m_j\zeta^j \right)\cdot (1,\zeta,\dots,\zeta^{n-1})\trans.
\]
\end{proposition}

\begin{proof}
Denote $\vec{v}=(1,\zeta,\dots,\zeta^{n-1})$ and $\vec{v}_0=(\vec{v} | -\vec{v})=(1,\zeta,\dots,\zeta^{2n-1})$.
The $2n\times 2n$ matrix
\[
M_0=\left(
\begin{array}{r|r}
M & -M\\
\hline
-M & M
\end{array}
\right)
\]
is a circulant matrix with the first row $(m_0,m_1,\dots,m_{n-1},-m_0,-m_1,\dots,-m_{n-1})$.  The vector $(1,\zeta,\dots,\zeta^{2n-1})\trans$ is a right eigenvector of $M_0$~\cite{Davis1979}:
\[
M_0\cdot \vec{v}_0 \trans =  \left(\sum_{j=0}^{n-1} m_j\zeta^j+\sum_{j=0}^{n-1} (-m_j)\zeta^{j+n} \right)\cdot \vec{v}_0\trans
= 2\left(\sum_{j=0}^{n-1} m_j\zeta^j\right)\cdot \vec{v}_0\trans.
\]
Writing $M_0$ and $\vec{v}_0$ in the block form using $M$ and $\vec{v}$ gives the result.

\end{proof}
In particular, note that all $n\times n$ pseudo-circulant matrices have the
common complex eigenvectors $(1,\zeta,\dots,\zeta^{n-1})\trans$ for all $\zeta\in\mathbb{C}$ that satisfy $\zeta^n=-1$.

The Sub Rosa expansion $\phi_n$ is determined by its edgeword, \ie  the sequence of edges and rhombuses on
the boundaries of the metatiles. In fact, the order of the letters in the edgeword is irrelevant:
the abelianized edge word determines $\phi_n$.

\begin{definition}[Abelianized edgeword $\abel{u}$]
Given an edgeword $u$ on the alphabet $A_n = \{0,2,4,6\dots n-2\}$, we define the abelianized edgeword $\abel{u}$ as the vector of number of occurrences of each symbol in the edgeword $u$, \ie
\[ \abel{u} := (|u|_a)_{a\in A_n} \]
\end{definition}
Remark that since the alphabet $A_n$ is the set of even numbers between 0 and $n-2$, this means that $\abel{u}$ is a vector on $n/2$ coordinates and that the $i$th coordinate of $[u]$ is the number of occurrences of the symbol $2i$, \ie $[u]_i = |u|_{2i}$.

From the definition of the Sub Rosa edgewords we get an exact formula for the abelianized Sub Rosa edge words $\abel{\Sigma(n)}$ for even $n$:
\[ \abel{\Sigma(n)} = \tuple{n-2i}_{0\leq i < \frac{n}{2}} = 2\tuple{\tfrac{n}{2} - i}_{0\leq i < \frac{n}{2}}. \]
We have, for example, $\abel{\Sigma(4)} = (4,2)$ as $\Sigma(4)=020020$, and we have $\abel{\Sigma(6)} = (6,4,2)$ as $\Sigma(6)=024020020420$.

\begin{proposition}[Expansion matrices]
  \label{prop:expansion_subrosa_even}
  For even $n$, the expansion matrix $\expansionMatrix{\phi}$ of a substitution of an edgeword $u$ is a pseudo-circulant matrix of the form

  \[\expansionMatrix{\phi} =  \begin{pmatrix} m_0 && -m_{n-1} && \dots && -m_1 \\ m_1 && m_0 && \dots && -m_2 \\ \vdots && \ddots && \ddots && \vdots \\ m_{n-1} && m_{n-2} && \dots && m_0 \end{pmatrix}\]
  with $m_0:= \abel{u}_0$, $m_{n/2}:= 0$,
  and $\forall 1\leq i < \tfrac{n}{2},\ m_{i} = \abel{u}_i,\ m_{ n - i } = -\abel{u}_i$.
\end{proposition}

\begin{proof}
  Let $\sigma$ be a substitution (on rhombus tilings with $n$ edge directions) of the edgeword $u$. In particular recall that it means that in the substitution the image of any two edges is the same up to translation and rotation.
  The edge of the metatiles of $\sigma$ is a succession of edges and rhombuses which are bisected tiles along an even angle (see Figure \ref{fig:subrosa_edge_vectors_even}) It means that if we take the edge $\vv{0}$ (or in $\mathbb{R}^n$ the edge $\ee{0}$) its image $\boundary\sigma(\vv{0})$ is a succession of edges $\vv{0}$ and rhombuses where the rhombus of angle $\tfrac{2\pi}{n}$ has edges $\vv{1}$ and $-\vv{n-1}$, and similarly the rhombus of angle $\tfrac{2k\pi}{n}$ has edges $\vv{k}$ and $-\vv{n-k}$ as in Figure \ref{fig:subrosa_edge_vectors_even}. In particular the number of vector $\vv{k}$ in $\boundary\sigma(\vv{0})$ is exactly the number of rhombuses of angle $\tfrac{2k\pi}{n}$ in the edgeword, which means that $m_k:= \abel{u}_k$. Remark that on the edge of the metatiles we have only rhombuses with even angles but they appear in both possible orientations.
\end{proof}
In particular, Proposition~\ref{prop:expansion_subrosa_even} 
applies to Sub Rosa substitutions and expansions with $u=\Sigma(n)$.
\begin{figure}
    \center
    \includegraphics[width=0.6\textwidth]{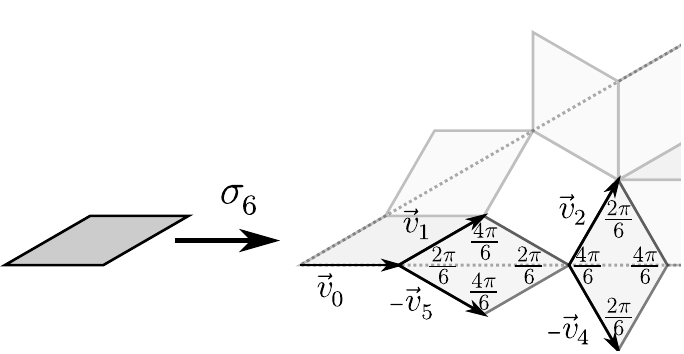}
    \caption{The edges and rhombuses on the edge of the metatile for even $n$.}
    \label{fig:subrosa_edge_vectors_even}
\end{figure}

\subsection{Decomposition of the expansion}

Propositions~\ref{prop:preudocirculant} and \ref{prop:expansion_subrosa_even} provide eigenvectors and eigenvalues
for the expansion of a substitution with a given edgeword $u$. To simplify the calculations of eigenvectors and eigenvalues we first decompose the expansion matrix into elementary components that correspond to single letter edgewords.

\begin{definition}[Elementary matrix $\elementaryMatrix{n}{i}$]
    \label{def:elementarymatrices}
    We define the elementary matrices $\elementaryMatrix{n}{i}$ of size $n\times n$ for $0\leq i< \frac{n}{2}$ by
    \begin{itemize}
        \item $\elementaryMatrix{n}{0}:= \identity{n}$,
        \item $\elementaryMatrix{n}{i}$ is the pseudo-circulant matrix with $m_i = 1$  and $m_{n-i} = -1$.
    \end{itemize}
\end{definition}
For example for $n=6$, there are three elementary matrices:
\begin{align}
\elementaryMatrix{6}{0} &:=
\begin{pmatrix}
    1 & 0 & 0 & 0 & 0 & 0 \\
    0 & 1 & 0 & 0 & 0 & 0 \\
    0 & 0 & 1 & 0 & 0 & 0 \\
    0 & 0 & 0 & 1 & 0 & 0 \\
    0 & 0 & 0 & 0 & 1 & 0 \\
    0 & 0 & 0 & 0 & 0 & 1
\end{pmatrix} \qquad
\elementaryMatrix{6}{1}:=
\begin{pmatrix}
    0 & 1 & 0 & 0 & 0 & -1 \\
    1 & 0 & 1 & 0 & 0 & 0 \\
    0 & 1 & 0 & 1 & 0 & 0 \\
    0 & 0 & 1 & 0 & 1 & 0 \\
    0 & 0 & 0 & 1 & 0 & 1 \\
    -1 & 0 & 0 & 0 & 1 & 0 \\
\end{pmatrix} \\
\elementaryMatrix{6}{2} &:=
\begin{pmatrix}
   0 & 0 & 1 & 0 & -1 & 0 \\
   0 & 0 & 0 & 1 & 0 & -1 \\
   1 & 0 & 0 & 0 & 1 & 0 \\
   0 & 1 & 0 & 0 & 0 & 1 \\
   -1 & 0 & 1 & 0 & 0 & 0 \\
   0 & -1 & 0 & 1 & 0 & 0 \\
\end{pmatrix}
\end{align}

\begin{proposition}[Decomposition of the expansion matrix]
    \label{prop:expansiondecomposition}
    Let $\expansionMatrix{\phi}$ be the expansion matrix of a substitution of edgeword $u$.
    We can decompose $\expansionMatrix{\phi}$ in a linear combination of the elementary matrices:
    \[ \expansionMatrix{\phi} = \sum\limits_{0\leq i < \frac{n}{2}} \abel{u}_i \elementaryMatrix{n}{i}.\]
\end{proposition}
\begin{proof}
    This is a direct consequence of Proposition \ref{prop:expansion_subrosa_even} and Definition \ref{def:elementarymatrices}.

\end{proof}
This proposition applies to Sub Rosa expansion matrices with the abelianized edge-word $\abel{\Sigma(n)}$. For example, for $n=4$ we have
\[ \expansionMatrix{4}= 4\elementaryMatrix{4}{0} + 2\elementaryMatrix{4}{1},\]
and for $n=6$ we have
\[ \expansionMatrix{6}= 6\elementaryMatrix{6}{0} + 4\elementaryMatrix{6}{1} + 2\elementaryMatrix{6}{2}.\]
Proposition~\ref{prop:preudocirculant} immediately provides the following.

\begin{proposition}[Eigenspaces and eigenvalues of the elementary matrices]
\label{prop:elementaryeigenvalue}
Let $n$ be an even integer greater than 2, and let $0\leq i < \frac{n}{2}$.

For every $0\leq k < \tfrac{n}{2}$, the $i$'th elementary matrix $\elementaryMatrix{n}{i}$ of Definition \ref{def:elementarymatrices} admits the space $\slope_n^k$ real two-dimensional eigenspace
with the eigenvalue
\[ \elementaryEigenvalue{n}{i}{k}:= \begin{cases} 1 \qquad \qquad \qquad \qquad \text{if } i = 0, \\     2\cos(\frac{i(2k+1)\pi}{n}) \qquad \text{ if } 0 < i < \tfrac{n}{2}. \end{cases}\]

\end{proposition}

\begin{proof}
Let us first remark that for $i=0$ the statement is trivial because $\elementaryMatrix{n}{0} =\identity{n}$ and the identity matrix admits any space as an eigenspace with the eigenvalue 1.

Let us now take $0< i < \tfrac{n}{2}$ and $0\leq k < \tfrac{n}{2}$, and denote $\zeta=e^{\imag\frac{(2k+1)\pi}{n}}$.
Then $\zeta^n=-1$ and $\slope_n^k=\spanning{\vec{v}}$ with
$\vec{v}=(1,\zeta,\dots,\zeta^{n-1})$. By Proposition~\ref{prop:preudocirculant}
the vector $\vec{v}$ is an eigenvector of the pseudo-circulant $\elementaryMatrix{n}{i}$ with the
corresponding eigenvalue
\[
\zeta^i-\zeta^{n-i}= e^{\imag \frac{i(2k+1)\pi}{n}} - e^{\imag \frac{(n-i)(2k+1)\pi}{n}} = 2\cos\left(\frac{i(2k+1)\pi}{n}\right).
\]
\end{proof}

\subsection{Eigenvalues of the expansion}

\begin{definition}[The eigenvalue matrix]
\label{def:eigenvalue_matrix}
Let $n$ be an even integer, let us define the eigenvalue matrix $\eigenMatrix{n}$ by
\[ \eigenMatrix{n}:= \left(\elementaryEigenvalue{n}{j}{i} \right)_{0\leq i,j < \frac{n}{2}} = \left(\eta_j \cos \tfrac{(2i+1)j\pi}{n}\right)_{0\leq i,j < \frac{n}{2}},\]
with $\eta_0:=1$ and $\eta_j:= 2$ for $0<j<\tfrac{n}{2}$.
\end{definition}

\begin{proposition}[Eigenvalues]
\label{prop:eigenvalue}
Let $n$ be an even integer.
Let $\sigma$ be a substitution, on rhombus tiles with $n$ edge directions, such that the image of any two edges by $\sigma$ are the same up to rotation and translation. Let $u$ be the edgeword of $\sigma$, and let $\expansionMatrix{\phi}$ be its expansion matrix.

For every $0\leq k < \frac{n}{2}$, the matrix $\expansionMatrix{\phi}$ admits the space $\slope_n^k$ as an eigenspace.
Letting $\lambda_{k}$ be the corresponding eigenvalue, and denoting $\lambda:= (\lambda_{k})_{0\leq k < \frac{n}{2}}$
for the vector of these eigenvalues, we have
\[ \lambda\trans = \eigenMatrix{n}\cdot \abel{u}\trans.\]
\end{proposition}
\begin{proof}
This proposition follows from Propositions \ref{prop:expansiondecomposition} and \ref{prop:elementaryeigenvalue}.
Let us take $n$ even, and a substitution $\sigma$ of edge word $u$ and expansion matrix $\expansionMatrix{\phi}$.
We decompose $\expansionMatrix{\phi}$ as
\[ \expansionMatrix{\phi} = \sum\limits_{0\leq i < \frac{n}{2}} \abel{u}_i \elementaryMatrix{n}{i}.\]
Now take $0\leq k < \tfrac{n}{2}$.
All the elementary matrices admit $\slope_n^k$ as eigenspace. So $\expansionMatrix{\phi}$ also admits $\slope_n^k$ as eigenspace. Let us denote by $\lambda_{k}$ the eigenvalue of $\expansionMatrix{\phi}$ on $\slope_n^k$. We have
\[ \lambda_{k} = \sum\limits_{0\leq i < \frac{n}{2}} \abel{u}_i \elementaryEigenvalue{n}{i}{k}.\]
Overall we can reformulate this as
\[ \lambda\trans=\eigenMatrix{n}\cdot \abel{u}\trans.\]

\end{proof}

\begin{proposition}[Eigenvalues of the Sub Rosa expansion]
\label{prop:eigenvalues_subrosa}
Let $n$ be an even integer. Let $\expansionMatrix{n}$ be the expansion matrix of Sub Rosa substitution $\sigma_n$.
For all $0\leq k < \tfrac{n}{2}$:
\begin{enumerate}
\item $\expansionMatrix{n}$ admits $\slope_n^k$ as an eigenspace with some eigenvalue which we denote by $\lambda_{n,k}$,

\item $\lambda_{n,k} = \frac{1}{\sin^2\left(\frac{(2k+1)\pi}{2n}\right)} $,

\item $\lambda_{n,k} > 1$.
\end{enumerate}
\end{proposition}

\begin{proof}
Fix $0\leq k < \tfrac{n}{2}$ and  denote $\theta:= \tfrac{(2k+1)\pi}{2n}$. Note that $0<\theta < \tfrac{\pi}{2}$.
\medskip

\noindent
1. The first item is a direct consequence of Proposition \ref{prop:eigenvalue} with the fact that $\sigma_n$ is defined by an edgeword $\Sigma(n)$.
\medskip

\noindent
2. Because $\abel{\Sigma(n)} = 2\tuple{\tfrac{n}{2} - i}_{0\leq i < \frac{n}{2}}$, we get from
Proposition \ref{prop:eigenvalue} that
\[ \lambda_{n,k} = \sum\limits_{0\leq i < \frac{n}{2}} \abel{\Sigma(n)}_i \lambda_{n,k,i} = \sum\limits_{0\leq i < \frac{n}{2}} 2(\tfrac{n}{2}-i) \elementaryEigenvalue{n}{i}{k}. \]
Replacing $\elementaryEigenvalue{n}{i}{k}$ with its value from Definition~\ref{def:eigenvalue_matrix} we get
\[ \lambda_{n,k} = n + 4\sum\limits_{1\leq i < \frac{n}{2}} \left(\tfrac{n}{2}-i\right)\cos\left(2i\theta\right).\]
We can calculate

\begin{align}
&\lambda_{n,k}\sin^2(\theta)\\
&= \left(  n + 4\sum\limits_{1\leq i < \frac{n}{2}} \left(\tfrac{n}{2}-i\right)\cos\left(2i\theta\right) \right) \sin^2(\theta) \\
&= \left(  n + 4\sum\limits_{1\leq i < \frac{n}{2}} \left(\tfrac{n}{2}-i\right)\cos\left(2i\theta\right) \right)\left(\frac{1-\cos(2\theta)}{2}\right)\\
&= \tfrac{n}{2}(1-\cos(2\theta)) + \sum\limits_{1\leq i < \frac{n}{2}}(\tfrac{n}{2}-i)\Big( 2\cos(2i\theta) - \cos(2(i+1)\theta) - \cos(2(i-1)\theta)\Big)\\
&= \tfrac{n}{2}(1-\cos(2\theta)) + \tfrac{n}{2}\cos(2\theta)-\cos\left(n\theta\right) - (\tfrac{n}{2} -1) \\
&= 1 - \cos\left(\frac{(2k+1)\pi}{2}\right) = 1 - \cos\left(k\pi+\frac{\pi}{2}\right)\\
&= 1 \end{align}
From this we get the expected formula
\[ \lambda_{n,k} = \frac{1}{\sin^2(\theta)} = \frac{1}{\sin^2\left(\frac{(2k+1)\pi}{2n}\right)}.\]
\medskip

\noindent
3. The third item now follows directly from the second item: we have $0<\theta < \tfrac{\pi}{2}$, so that $0< \sin^2\theta < 1$.
This implies that $\lambda_{n,k} > 1$.
\end{proof}

\begin{table}
\qquad
    \begin{tabular}{l| c c c c c c}
      n & $\lambda_{0}(n)$ & $\lambda_{1}(n)$ & $\lambda_{2}(n)$ & $\lambda_{3}(n)$ & $\lambda_{4}(n)$ & $\lambda_{6}(n)$\\
      \hline
      4 & 6.83 & 1.17 & - & - & - & -\\
      6 & 14.93 & 2 & 1.07 & - & - & - \\
      8 & 26.27 & 3.24 & 1.45 & 1.04 & - & - \\
      10 & 40.86 & 4.85  & 2  & 1.26 & 1.03 & -\\
      12 &  58.70 & 6.83 & 2.70  & 1.59 & 1.17 & 1.02 \\
      \end{tabular}

\caption{Approximated eigenvalues of the Sub Rosa expansion $\phi_n$ for small $n$.}
\label{table:approximate_eigenvalues}
\end{table}

Approximate values of the eigenvalues of the Sub Rosa expansions $\phi_n$ for small $n$ can be found in Table \ref{table:approximate_eigenvalues}.

In particular when we combine Proposition \ref{prop:eigenvalues_subrosa} and Proposition \ref{prop:planarity_non-planarity} we obtain Theorem \ref{th:subrosa_not_planar} for even $n\geq 4$. Indeed, the Sub Rosa expansion $\phi_n$ for even $n\geq 4$ has all eigenvalues strictly greater than 1 so the Sub Rosa substitution $\sigma_n$ is not planar, \emph{i.e.}, it does not admit discrete plane tilings. So the Sub Rosa $n$ tilings are not discrete planes.

\section{Planar Rosa substitution discrete plane tilings}
\label{sec:planar-rosa}
In this section we present the Planar Rosa construction for even $n$ and we prove Theorem \ref{th:planar-rosa} for even $n$. Note that the name Planar Rosa was chosen because this construction is adapted from the Sub Rosa construction to obtain the additional discrete plane condition.
Recall that the Planar Rosa construction and proof of Theorem \ref{th:planar-rosa} for odd $n$ can be found in \cite{kari2021}.

Let us fix an even integer $n\geq 4$ throughout this section.
We consider here tilings with unit rhombus tiles on $n$ edge directions as described in Section \ref{sec:settings}, and we consider vertex hierarchic substitutions $\sigma$ such that the image of any two edges by $\boundary\sigma$ are the same up to rotation and translation. This means that our substitutions $\sigma$ are defined by a palindromic
edgeword $u$ together with a way to tile the interior of the metatiles given their boundary.

\subsection{Construction}
\label{subsec:prConstruction}
\begin{definition}[Optimal frequency vector $\gamma$]
Let us define a vector $\gamma\in\R{\nhalf}$, which we call optimal frequency vector, as
\[ \gamma:= \tuple{\cos(\tfrac{i\pi}{n})}_{0\leq i < \frac{n}{2}}.\]
\label{def:gamma}
\end{definition}

We call $\gamma$ the optimal frequency vector because, as is be detailed in Subsection \ref{subsec:prPlanarity}, an edgeword with a proportion of each type of rhombus close enough to $\gamma$ guarantees the planarity of the corresponding substitution. Next we define a particular infinite word in which the letters indicating rhombus types appear balanced in the ratio given by the optimal frequency vector $\gamma$.

\begin{definition}[Billiard word $\bword$]
Let $\Gamma_{\half}\subseteq \R{\nhalf}$ be the line $\spanning{\gamma}+\tuple{\thalf, \dots \thalf}$ where
\[ \spanning{\gamma}+\tuple{\thalf, \dots \thalf}:= \left\{ t\gamma + \tuple{\thalf, \dots \thalf}, t \in \mathbb{R} \right\}.\]

Let $\bword$  be the one-way-infinite billiard word of line $\Gamma_\half$ starting at $\tuple{\thalf,\dots \thalf}$, meaning that $\bword$ is built by travelling along the line and adding a letter $2i$ to the word each time it crosses a hyperplane $\hyperplane{i}{k}$ with $\ee{i}$ a vector of the canonical basis of $\R{\nhalf}$, $k\in\mathbb{N}$ and
\[ \hyperplane{i}{k}:= \{ x \in \R{\nhalf}| x\cdot \ee{i} = k \}. \]
\end{definition}
Remark that though the optimal frequency vector is not always totally irrational (for $n=6$, for example, it is not totally irrational) the billiard word $\bword$ is well-defined , see Lemma \ref{lemma:billiard_word_well_defined} in the Appendices for a proof of that.
Note that in our definition, when crossing an hyperplane $\hyperplane{i}{k}$ we add a letter $2i$ to the word (instead of a letter $i$ as is usual with billiard words), this is because the word $\bword$ is thought of as an infinite edgeword and the letters represent rhombuses and edges so they are even integers between $0$ and $n-2$.

Next we use this (infinite) billiard word to define arbitrarily long palindromes.

\begin{definition}[Candidate edgeword $\prWord{i}$ and candidate substitution $\prSubs{i,n}$]
~\\
Let us define for all $i\in\N$ the  candidate edgeword
\[ \prWord{i}:= pref_i(\bword)\overline{pref_i(\bword)},\]
where $pref_i(\bword)$ is the prefix of length $i$ of the billiard word $\bword$ defined above, and $\overline{pref_i(\bword)}$ is its mirror image.

If the metatiles of angles $(\tfrac{k\pi}{n}, \tfrac{(n-k)\pi}{n})$ with edgeword $\prWord{i}$ are tileable for all $1\leq k \leq \nhalf$ then we define the candidate substitution $\prSubs{i,n}$ as a substitution with edgeword $\prWord{i}$.
\label{def:candidate}
\end{definition}

\begin{figure}[t]
\includegraphics[width=\textwidth]{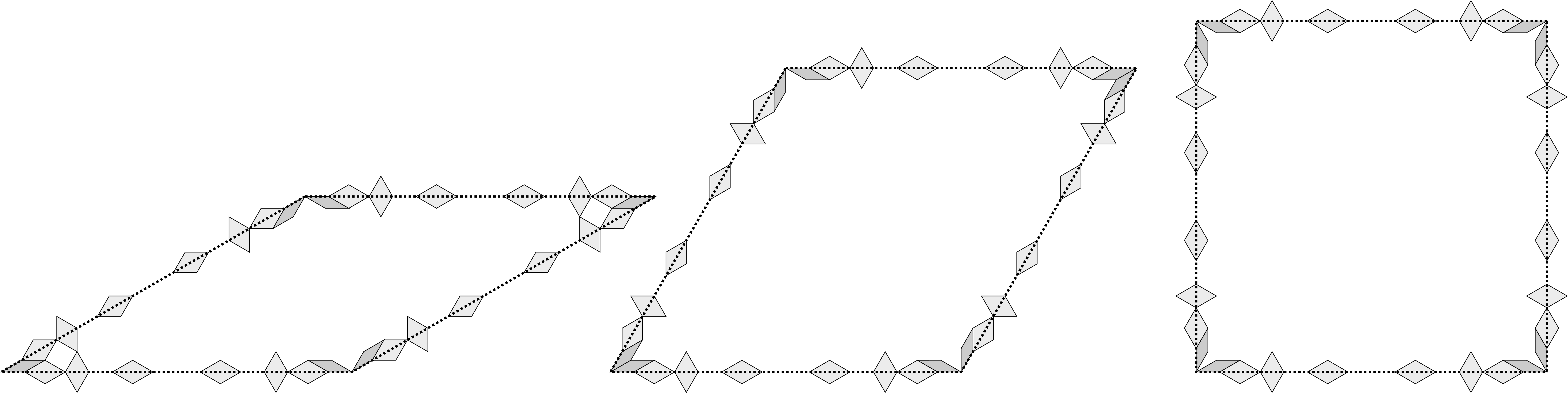}
\caption{The corner condition for $n=6$: in every corner of the three metatiles there is a $\tfrac{\pi}{6}$ rhombus on both sides.}
\label{fig:corner-condition}
\end{figure}

\begin{definition}[Corner condition on the Planar Rosa metatiles]
\label{def:corner-condition}
We say that a substitution $\sigma$ satisfies the corner condition when in every corner of every metatile of $\sigma$ there is a narrow rhombus (rhombus of angle $\tfrac{\pi}{n}$) on both sides, see Figure \ref{fig:corner-condition}.
\end{definition}

\begin{figure}[b]
\center
\includegraphics[width=0.5\textwidth]{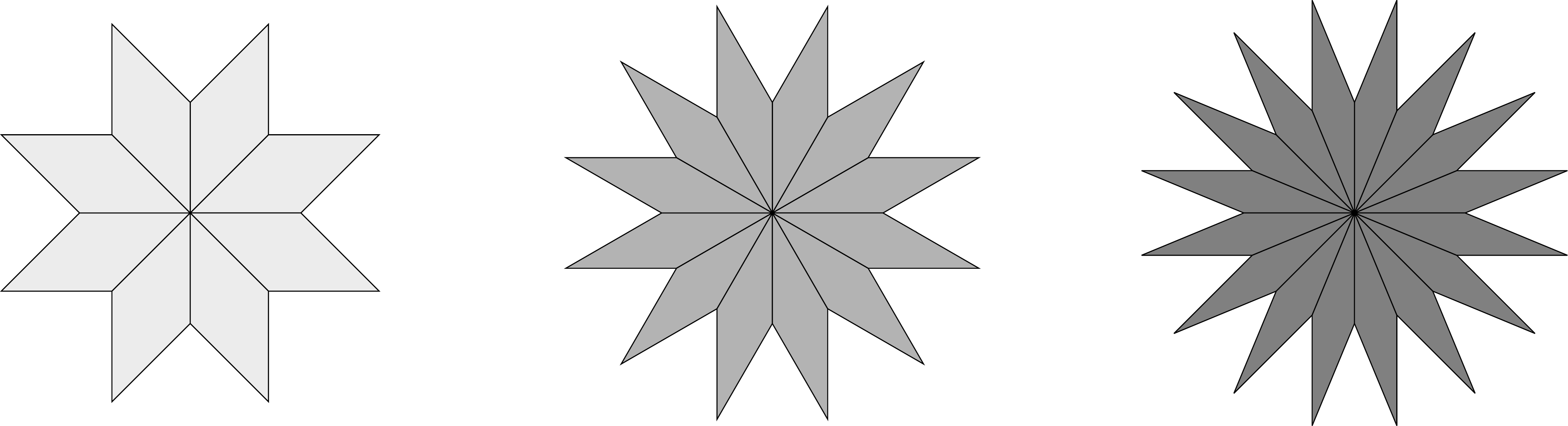}
\caption{The star pattern $\star{n}$ for $n\in\{4,6,8\}$.}
\label{fig:star}
\end{figure}

\begin{definition}[Star pattern $\star{n}$]
Let $n$ be an even integer.
We call star pattern denoted by $\star{n}$ the pattern consisting of $2n$ rhombuses of angle $\tfrac{\pi}{n}$ around a single vertex, see Figure \ref{fig:star}.
\label{def:star}
\end{definition}

\begin{definition}[Planar Rosa substitution $\prSubs{n}$ and canonical tiling $\prTiling{n}$]
\label{def:planar-rosa}
We define the Planar Rosa substitution $\prSubs{n}$ as the candidate substitution $\prSubs{i,n}$ with the minimal $i$ such that it is a a primitive substitution that satisfies the corner condition and it is a planar substitution of slope $\slope_n^0$.

We define the canonical Planar Rosa tiling $\prTiling{n}$ as the fixpoint of $\prSubs{n}$ from the star pattern $\star{n}$.
\end{definition}
Remark that the definition of the Planar Rosa substitution $\prSubs{n}$ requires the existence of an integer $i$ such that the candidate substitution $\prSubs{i,n}$ exists (\ie the metatiles defined by the candidate edgeword $\prWord{i}$ are tileable), satisfies the corner condition, is planar of slope $\slope_n^0$ and is primitive. Additionally, the definition of the canonical Planar Rosa tiling $\prTiling{n}$ requires the fact that the star pattern $\star{n}$ is a legal seed for the Planar Rosa substitution $\prSubs{n}$. The rest of the paper concerns proving these facts.

In Section \ref{subsec:prTileability} we prove that for all sufficiently large $i$ the metatiles defined by the candidate edgeword $\prWord{i}$ are tileable in a way that satisfies the corner condition, \ie the candidate substitution $\prSubs{i,n}$ exists and satisfiest the corner condition.

In Section \ref{subsec:prPlanarity} we prove that for infinitely many $i$, provided that the candidate substitution $\prSubs{i,n}$ exists, the candidate substitution is planar of slope $\slope_n^0$, \ie it generates tilings that are discrete plane tilings of slope $\slope_n^0$.

In Section \ref{subsec:prPrimitivity} we prove that for all sufficiently large $i$, provided that the candidate substitution $\prSubs{i,n}$ exists, the candidate substitution is primitive.

In Section \ref{subsec:prSeed} we prove that for all $i$, provided that the candidate substitution $\prSubs{i,n}$ exists and satisfies the corner condition, the star pattern $\star{n}$ is a legal seed for the candidate substitution $\prSubs{i,n}$.

In Section \ref{subsec:prConclusion} we combine these results to prove that Planar Rosa substitution $\prSubs{n}$ as defined above exists, and thus we
prove Theorem \ref{th:planar-rosa}.

Note that the Planar Rosa substitution $\prSubs{n}$ defines a set of tilings called a \emph{subshift}. Since the substitution is primitive, all these tilings are locally indifferentiable but, as in the case of Penrose tilings, there are uncountably many tilings in the subshift. To define a single tiling that represents this whole subshift, we use a construction of fixpoint of the substitution from a legal seed and we call it the \emph{canonical} tiling.

\subsection{Tileability of the Planar Rosa metatiles}
\label{subsec:prTileability}

In this section we present the proof of the tileability of the metatiles defined by the candidate edgewords.

The main result of this section is the following proposition.
\begin{proposition}[Tileability]
\label{prop:prTileability}
There exists $K\in\mathbb{N}$ such that for any $i\geq K$, the metatiles defined by the candidate edgeword $\prWord{i}$ are tileable in a way that satisfies the corner condition.
\end{proposition}
We first prove the tileability of the metatiles, without taking into account the corner condition.
In order to prove the tileability we need to introduce the Kenyon criterion and the counting functions induced by an edgeword, their link with tileability and results on the multigrid dual tilings. The actual proof of Proposition \ref{prop:prTileability} is decomposed in two parts: first without the corner condition, and then adding the corner condition.

The interior of the metatile defined by a candidate edgeword $\prWord{i}$ is a polygon (with unit length edges, see Figure \ref{fig:kenyon_criterion}) and our goal is now to tile it with parallelograms (rhombuses in our case).

\begin{figure}[b]
\includegraphics[width=\textwidth]{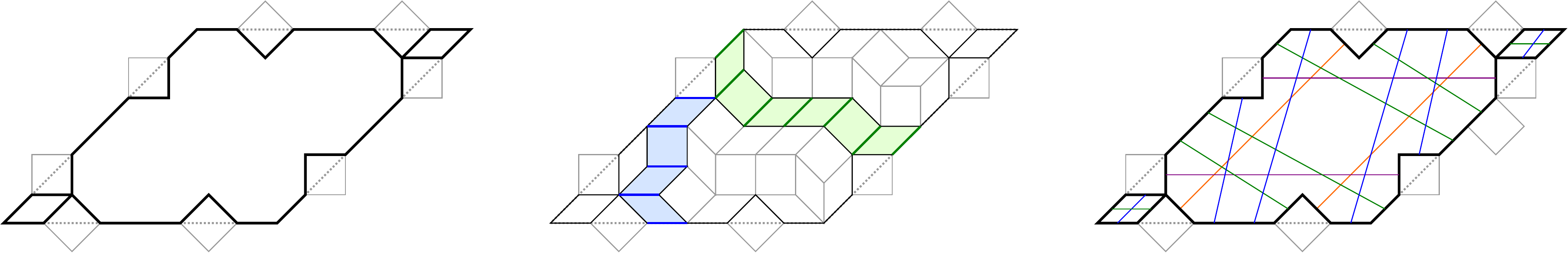}
\caption{The Kenyon criterion for the candidate edgeword $\prWord{3}=020020$ for $n=4$: on the left the polygon to tile, in the middle the chains of rhombus in a tiled metatile, on the right the Kenyon matching.}
\label{fig:kenyon_criterion}
\end{figure}

The first step of the Kenyon method is to fix a starting vertex on the polygon to tile. Let us write $(\vec{a}_1,\dots,\vec{a}_m)$ for the sequence of oriented edges when going around the boundary of the polygon counterclockwise from the starting point and back. All $\vec{a}_j$ are in the set of directions of the tiling, \ie, $\vec{a}_j\in \{\pm \vec{v}_k\ |\ k=1,\dots ,n\}$.
We say that $\vec{v}_k$ and $-\vec{v}_k$ have the same \emph{edge type} but opposite directions.
We denote $\vec{a}_j^\bot$ for the unit vector orthogonal to $\vec{a}_j$ in the counterclockwise direction.

Suppose the interior of the polygon is actually tiled. As seen in Figure \ref{fig:kenyon_criterion}, from each edge of the polygon starts a chain of rhombuses (highlighted in Figure \ref{fig:kenyon_criterion}) that share the same edge type. At the two ends of this chain there
are two edges of the polygon with the same edge types but opposite directions.
These chains define a matching of pairs of edges of the polygon with the following properties.
\begin{enumerate}
  \item[K1.] Two edges that are matched have the same edge type but opposite orientations.
  \item[K2.] Two matched pairs of edges of the same edge types cannot cross each other with respect to the cyclic ordering of the edges. Indeed, two chains with the same edge type cannot cross, as such a crossing would create a ``flat'' rhombus.
  \item[K3.] Two matched edges must ``see'' each other in the parallelogram, \ie there is a path in the interior of the polygon from one edge to the other that is monotonically increasing in the direction $\vec{a}_j^\bot$.
  Indeed, otherwise it would mean that the corresponding chain of rhombuses, linking the two edges, circles back on itself.
  \item[K4.] The matching is peripherally monotonous: for any two matched pairs $\{\vec{a},\vec{a}'\}$ and $\{\vec{b},\vec{b}'\}$ such that in the cyclic ordering of the edges $\vec{a}<\vec{b}<\vec{a}'<\vec{b}'$, we have $\vec{a}^\bot \cdot \vec{b}\trans > 0$. That ensures that at the crossing of the two chains a real rhombus actually exists.
\end{enumerate}
We call \emph{Kenyon matching} a matching on the oriented edges that follow these properties.
\begin{theorem}[Kenyon criterion \cite{kenyon1993}]
\label{thm:kenyon}
A polygon is tileable by parallelograms if and only if a Kenyon matching exists.
\end{theorem}

\begin{figure}[!b]
\center
\includegraphics[width=0.6\textwidth]{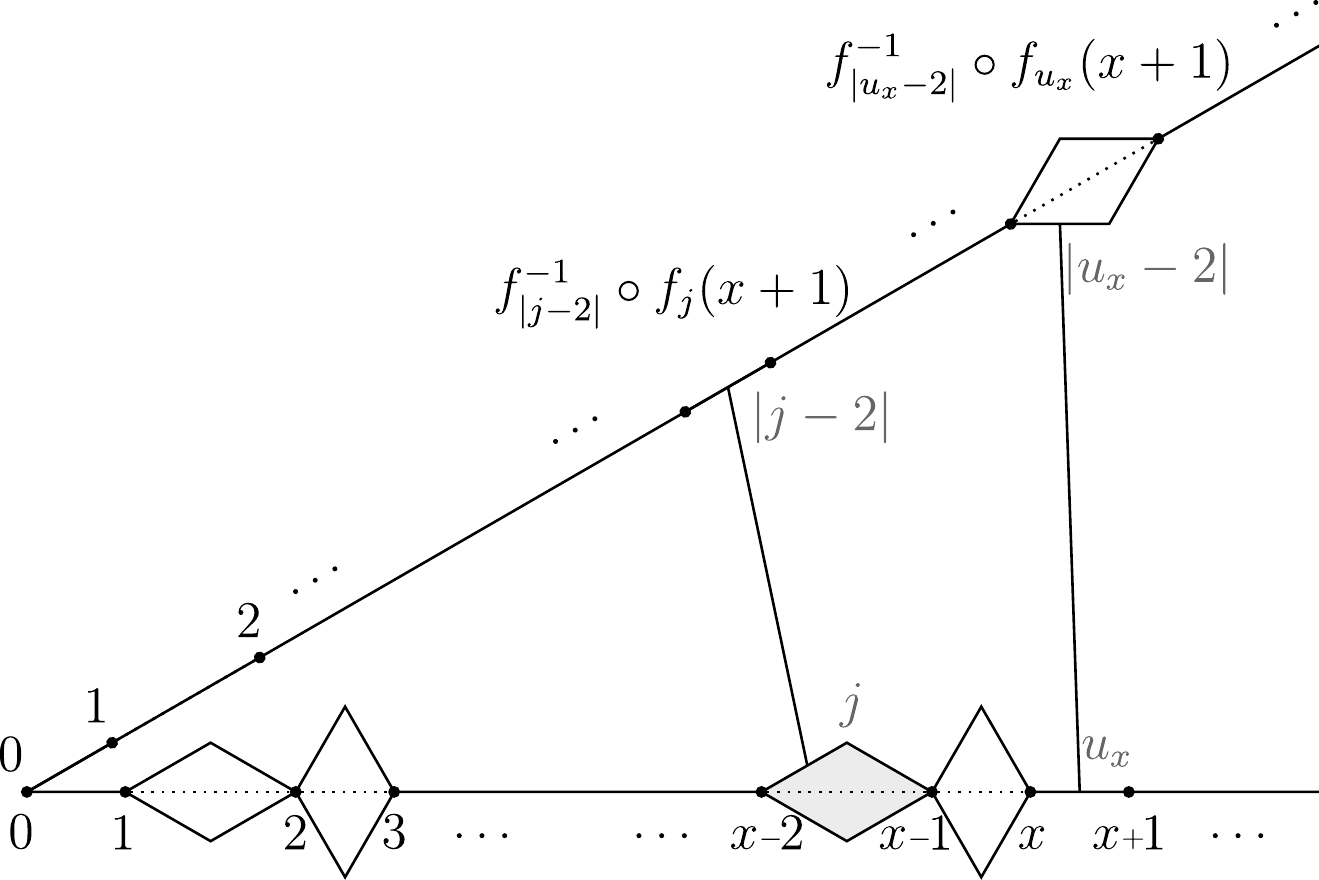}
\caption{The counting functions at the most narrow corner on Planar Rosa 6.
The edge at position $x$ is matched to a rhombus of type $|u_x-2|$ at position $f_{|u_x-2|}^{-1}\circ f_{u_x}(x+1)$. The last occurrence of the rhombus of type $j$ before position $x$ is matched to an edge or rhombus of type $|j-2|$ at position $f_{|j-2|}^{-1}\circ f_{j}(x+1)$. If $u_x < j$, these two matchings cannot cross otherwise we would have an invalid Kenyon crossing, therefore $f_{|j-2|}^{-1}\circ f_{j}(x+1) < f_{|u_x-2|}^{-1}\circ f_{u_x}(x+1)$ is a necessary condition for tileability.}
\label{fig:tileability}
\end{figure}

\begin{definition}[Counting functions]
Let $n\geq 4$ be an even integer and let $u=u_0u_1\dots u_{m-1}$ be an edgeword of length $m$ with letters in the alphabet $\{0,2,\dots n-2\}$.

For an even integer $j \in \{0,2,4,\dots n-2\}$ and a length $x\in\{0,\dots m\}$ we define $f_j(x)$ as the number of occurrences of the letter $j$ in the prefix of length $x$ of the edgeword $u$, \ie
\[f_j(x):= |u_0\dots u_{x-1}|_j.\]
We also define $f_j^{-1}(y)$ as the length of the shortest prefix of $u$ containing at least $y$ occurrences of the letter $j$, \ie
\[f_j^{-1}(y):= \inf\{x\in \{0,\dots m\} |\ f_j(x)\geq y\}\]
\end{definition}
We use these counting functions because we can reformulate the Kenyon criterion for the tileability of a polygon by parallelograms \cite{kenyon1993} in terms of inequalities on the counting functions.
In order to have a cleaner result we add an additional definition: the almost-balancedness of edgewords.

\begin{definition}[Almost-balancedness] Let $n\geq 4$ be an even integer and $k$ a positive integer.\\
We say that a word $u$ on alphabet $\{0,2,4,\dots n-2\}$ is $k$-almost-balanced when for any letters $j_1$, $j_2$ such that the frequency of appearance of $j_1$ in $u$ is at least that of $j_2$, and any factor subword $v$ of $u$ we have $|v|_{j_1} - |v|_{j_2} \geq -k$.
\end{definition}
The idea of almost-balancedness is that when the frequency of appearance of a letter $j_1$ is greater than the frequency of appearance of a letter $j_2$ in the word, then we have a lower bound on $|v|_{j_1}-|v|_{j_2}$ for all finite factor subwords $v$ of $u$.

Note that in the infinite billiard word $w$ defined in Section \ref{sec:prConstruction}, the frequency of appearance of the letters is decreasing: the frequency of appearance of $0$ is greater than that of $2$, which in turn is greater than that of $4$ , etc. So if $j_1 < j_2$ then in $w$ (and in large enough candidate edgewords $\prWord{i}$) the frequency of appearance of $j_1$ is greater than that of $j_2$.

\begin{proposition}[Tileability and counting functions \cite{kari2021}]
\label{prop:tileability_functions}
Let $n\geq 4$ be an even integer and let $u$ be a $2$-almost-balanced edgeword on alphabet $\{0,2,\dots n-2\}$ such that each letter appears at least once in $u$. Suppose that for every position $0\leq x < n$, and for every even $j$ such that $u_x < j < n$, we have
\begin{equation}
    \label{equation:counting_functions_criterion}
    f_{|j-2|}^{-1}\circ f_{j}(x+1) < f_{|u_x-2|}^{-1}\circ f_{u_x}(x+1).
\end{equation}
Then all the metatiles induced by the edgeword $u$ are tileable, implying that the edgeword $u$ actually induces a well-defined substitution.

Conversely if there exists $(x,u_x,j)$ that breaks this condition then the narrow metatile of edgeword $u$ is not tileable.
\end{proposition}
The full proof can be found in \cite{kari2021}, or in \cite{lutfalla2021these} with more details.
\begin{proof}[Proof sketch.]
The proof of this result is based on the equivalence between a triplet $(x,u_x,j)$ that breaks this condition and an invalid Kenyon crossing. Which means that, in our specific case of metatiles defined by an edgeword, the Kenyon criterion \cite{kenyon1993} is equivalent to our criterion on counting functions.
See Figure \ref{fig:tileability} for a sketch of the proof that our criterion on counting functions is a necessary condition for tileability, the proof that it is also a sufficient condition is harder.
\end{proof}
Before returning to the proof of Proposition \ref{prop:prTileability} let us give one last lemma.

\begin{lemma}[Infinite cone]
\label{lemma:infinite_cone}
The infinite cone of angle $\tfrac{\pi}{n}$ and the billiard word $\bword$ as the edgeword along both sides is tileable with unit rhombus tiles.
\end{lemma}
\begin{proof}
The proof of this lemma is based on multigrid dual tilings \cite{debruijn1981}. For a full definition of De Bruijn multigrids and their dual tilings and for the proof of the fact that for any $n$ the multigrid dual tiling $\dualtiling{n}{\tfrac{1}{2}}$ is a rhombus tiling we refer the reader to \cite{lutfalla2021multigrid}. For the rest of this proof we assume that the reader is familiar with multigrids and we use the fact that the multigrid $G_n(\tfrac{1}{2})$ is regular which means that its dual tiling $\dualtiling{n}{\tfrac{1}{2}}$ is a rhombus tiling.

In the multigrid dual tiling $\dualtiling{n}{\tfrac{1}{2}}$ the billiard word $\bword$ of rhombuses and edges appears on the horizontal half-line $\mathbb{R}^+$. Indeed, in the multigrid $\multigrid{n}{\tfrac{1}{2}}$ the horizontal half line $\mathbb{R}^+$ crosses either vertical grid lines (which correspond in the dual tiling to edges on the horizontal half line) or crossing points of two grid lines $\grid{\zeta^i}{\tfrac{1}{2}}\cap\grid{\zeta^{n-i}}{\tfrac{1}{2}}$ with $\zeta = e^{\imag\frac{\pi}{n}}$, which corresponds to a rhombus of angle $\tfrac{2i\pi}{n}$. Edges appears in positions $I_0$ with
\[I_0:= \{ k-\tfrac{1}{2},\ k> 1\},\]
and rhombuses of angle $\tfrac{2i\pi}{n}$ appear in position $I_i$ with
\[I_i = \left\{ \frac{k-\tfrac{1}{2}}{\cos(\frac{i\pi}{n})},\ k> 1 \right\}.\]
These positions are exactly the coordinates of the intersection of the line $\Gamma_{\frac{1}{2}}$ with the hyperplanes $H(i,k)$ (recall the Definition of $\bword$ as the billiard word of line $\Gamma_{\frac{1}{2}}$). So indeed $\bword$ appears on the horizontal half-line in $\dualtiling{n}{\tfrac{1}{2}}$.

By global $2n$-fold rotational symmetry of $\dualtiling{n}{\tfrac{1}{2}}$ we get that $\bword$ also appears on the rotated half line $\zeta \mathbb{R}$. So the infinite cone of angle $\tfrac{\pi}{n}$ and the edge word $\bword$ on both sides indeed appears in the tiling $\dualtiling{n}{\tfrac{1}{2}}$ which means that this infinite cone is tileable with unit rhombus tiles.
\end{proof}

Let us now use Proposition \ref{prop:tileability_functions} to prove the tileability of the metatiles induced by candidate edgewords, without considering the corner condition.
\begin{proof}[Proof of Proposition \ref{prop:prTileability} part 1: tileability]
\label{proof:prTileability_1}
We prove that for sufficiently large $i$, the candidate edgewords $\prWord{i}$ satisfy the conditions of Proposition \ref{prop:tileability_functions} so that the metatiles of edgeword $\prWord{i}$ are tileable and the candidate substitution $\prSubs{i,n}$ exists.

First we prove that the candidate edgewords are 2-almost-balanced.
Remark that $\bword$ is 1-almost-balanced. Indeed, take any $j_1<j_2$. Since almost-balancedness considers only one factor subword at a time we can first project $\bword$ on the alphabet $\{j_1,j_2\}$ just by deleting all the other letters. This projection is a billiard word on two letters, \ie a Sturmian word. This means that it is balanced and hence
also 1-almost-balanced.
Now recall that $\prWord{i} = pref_i(\bword)\overline{pref_i(\bword)}$, so any factor subword $v$ of $\prWord{i}$ is of the form $v=v_0\cdot v_1$, with $v_0$ a factor of $\bword$ and $v_1$ a factor of $\overline{\bword}$. Since $\bword$ is 1-almost-balanced then $\prWord{i}$ is 2-almost-balanced, for every $i$.

We now prove the condition on the counting functions.
We actually prove two statements:
\begin{enumerate}
\item For any $u=\prWord{i}$, and for any position $x\leq i$, the condition (\ref{equation:counting_functions_criterion}) is satisfied.
\item There exists $K$ such that for any $u=\prWord{i}$ and for any position $x> K$ the condition (\ref{equation:counting_functions_criterion}) is satisfied.
\end{enumerate}
Once we prove these two statements then we obtain that for any $i>K$ the metatiles induced by $\prWord{i}$ are tileable. Indeed, for any $x\leq i$ the condition is satisfied, and for any $x> i$ we also have $x> K$ so the condition is also satisfied.

Let us now prove these two statements.
The first statement is a direct corollary to Lemma \ref{lemma:infinite_cone}.
Remark that, when considering the Kenyon matching of the infinite cone of angle $\tfrac{\pi}{n}$, an edge $0$ is matched to a rhombus of type $2$, a rhombus of type $2$ is matched to an edge $0$ and a rhombus of type $2k$ with $k>1$ is matched to a rhombus of type $2(k-1)$.
Recall that, by definition of $\bword$, for any prefix $pref_i(\bword)$, and any $0\leq k < \tfrac{n}{2}$ we have $|pref_i(\bword)|_{2k} \geq |pref_i(\bword)|_{2(k+1)}$.
This implies that any rhombus at position $x$ in $\bword$ is matched to a rhombus or edge at position $x'\leq x$, and only an edge at position $x$ can be matched to a position $x'>x$.
With $\prWord{i}=pref_i(\bword)\overline{pref_i(\bword)}$, this means that the validity of Kenyon matching crossings or the validity of the counting condition in $\prWord{i}$ and in $\bword$ are identical for all positions $x\leq i$.
Overall, if there exists a $i\in \mathbb{N}$ and a $x\leq i$ such that the counting condition is broken with edgeword $\prWord{i}$ at position $x$ then the counting condition is also broken with the infinite edgeword $\bword$ at position $x$ which is impossible because the infinite cone of edgeword $\bword$ is tileable.

Let us now prove the second statement, \ie prove that there exists $K\in \mathbb{N}$ such that, for any integer $i$ and for any position $x> K$, the condition (\ref{equation:counting_functions_criterion}) is satisfied.
For this, we define $g_{j_1,j_2}(x)$ as
\[ g_{j_1,j_2}(x):= f_{|j_1-2|}^{-1}\circ f_{j_1}(x) - f_{|j_2-2|}^{-1}\circ f_{j_2}(x).\]
The idea is that for $j_1<j_2$, $g_{j_1,j_2}(x)$ has a general increasing trend and (if $\prWord{i}$ is long enough) there is $K_{j_1,j_2}$ such that for any $x> K_{j_1,j_2}$ holds $g_{j_1,j_2}(x)>0$.
Let us recall that with $\gamma = (\cos(\tfrac{i\pi}{n}))_{0\leq i < \frac{n}{2}}$ we have $f_{2i}(x)\approx x\cdot \tfrac{\gamma_i}{\| \gamma \|_1}$ by the construction of the word \bword.
For simplicity, we define $\tgamma{2i}:= \tfrac{\gamma_i}{\|\gamma\|_1}$, so now $f_i(x)\approx x\cdot \tgamma{i}$. This approximation is actually quite good in the sense that there exists $\delta$ dependent only on $n$ such that for each $i$ and $x$ we have $|f_i(x) - x\tgamma{i}| < \delta$. This means that
\[ \left|f_{|j_1-2|}^{-1}\circ f_{j_1}(x) - x\frac{\tgamma{j_1}}{\tgamma{|j_1-2|}} \right| < \delta + \frac{\delta}{\tgamma{|j_1-2|}}, \]
so that
\[ \left|g_{j_1,j_2}(x) - x\left( \frac{\tgamma{j_1}}{\tgamma{|j_1-2|}} - \frac{\tgamma{j_2}}{\tgamma{|j_2-2|}}\right) \right| < 2\delta + \frac{\delta}{\tgamma{|j_1-2|}} + \frac{\delta}{\tgamma{|j_2-2|}}.\]

An easy calculation shows that the sequence $\tfrac{\tgamma{0}}{\tgamma{2}}, \tfrac{\tgamma{2}}{\tgamma{0}}, \tfrac{\tgamma{4}}{\tgamma{2}}, \tfrac{\tgamma{6}}{\tgamma{4}},\dots$ is strictly decreasing, indeed $\tgamma{2i} = \cos(\tfrac{i\pi}{n}) / \|\gamma\|_1$ and the function $x\mapsto \cos(x+\tfrac{\pi}{n})/\cos(x)$ is decreasing on $[-\tfrac{\pi}{n}, \tfrac{\pi}{2}-\tfrac{\pi}{n}]$,  so that $\tfrac{\tgamma{j_1}}{\tgamma{|j_1-2|}} - \tfrac{\tgamma{j_2}}{\tgamma{|j_2-2|}}$ is positive when $j_1<j_2$.
Thus, for
\[ x > \frac{2\delta + \frac{\delta}{\tgamma{|j_1-2|}} + \frac{\delta}{\tgamma{|j_2-2|}}}{\frac{\tgamma{j_1}}{\tgamma{|j_1-2|}} - \frac{\tgamma{j_2}}{\tgamma{|j_2-2|}}}, \]
we have $g_{j_1,j_2}(x)>0$. Now take
\[ K:= \max\limits_{0\leq j_1 < j_2 < n} \frac{2\delta + \frac{\delta}{\tgamma{|j_1-2|}} + \frac{\delta}{\tgamma{|j_2-2|}}}{\frac{\tgamma{j_1}}{\tgamma{|j_1-2|}} - \frac{\tgamma{j_2}}{\tgamma{|j_2-2|}}}. \]
For any $x>K$ and for any $j_1< j_2$, we have $g_{j_1,j_2}(x) > 0$ which means that the inequality (\ref{equation:counting_functions_criterion}) holds.
So $\prWord{i}$ induces a well-defined substitution $\prSubs{i,n}$ for all sufficiently large $i$.
\end{proof}

Now that we proved that the metatiles induced by the candidate edgeword $\prWord{i}$ are tileable for all sufficiently large integer $i$ we prove that the corner condition can also be satisfied for all sufficiently large integer $i$.

\begin{proof}[Proof of Proposition \ref{prop:prTileability} part 2: the corner condition]
\label{proof:prTileability_2}
We consider an edgeword $\prWord{i}$ that induces tileable metatiles.
The question is now, can the metatiles be tiled in a way that satisfies the corner condition, \ie with a $\tfrac{\pi}{n}$ rhombus on both sides of all four corners of the metatile?
Let us take $1\leq k < n$ and let us look at the corner of angle $\tfrac{k\pi}{n}$ in the metatile of angles $\tfrac{k\pi}{n}$ and $\tfrac{(n-k)\pi}{n}$.

\begin{figure}[t]
\center
\includegraphics[width=0.5\textwidth]{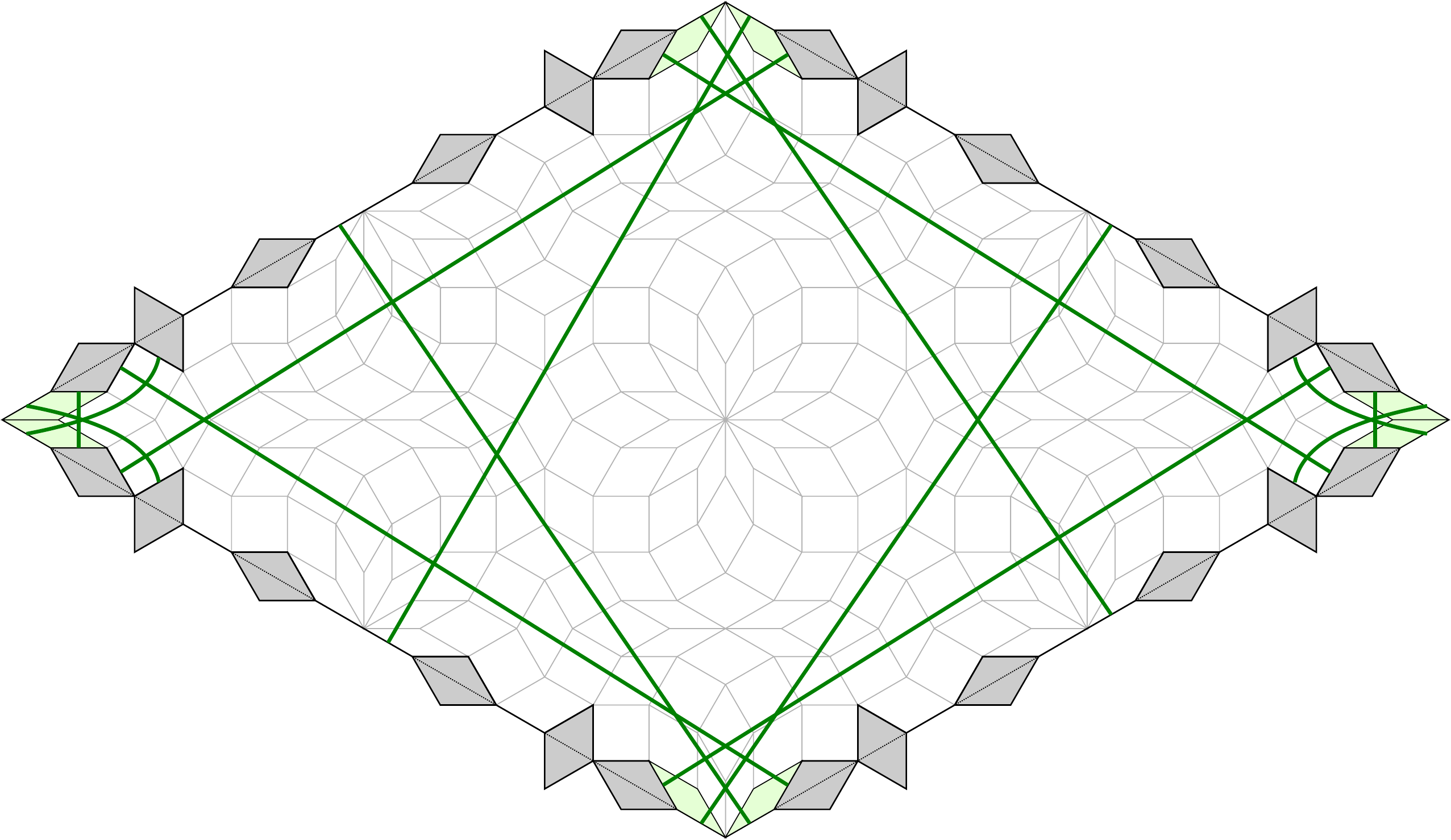}
\caption{For $n=6$ and for the candidate edgeword $\prWord{6}$, the metatile of angles $(\tfrac{\pi}{3}, \tfrac{2\pi}{3})$ can be tiled in such a way that satisfies the corner condition because, in every corner, the Kenyon matchings of the first edge and first rhombus cross.}
\label{fig:corner_crossing}
\end{figure}

Recall that the edgeword $\prWord{i}$ admits $02$ as a prefix, \ie it starts with an unit edge and by a $\tfrac{2\pi}{n}$ rhombus.
So the metatile can be tiled with a $\tfrac{\pi}{n}$ rhombus in the corner if and only if the chains of rhombuses starting from the first edge (letter 0 in the word) and the first rhombus (letter 2 in the word) cross, \ie if the Kenyon matchings of the first edge and the first rhombus cross, see Figure \ref{fig:corner_crossing}.
Let us also assume that $i$ is such that there it at least 1 occurrence of each letter in $pref_i(\bword)$.

There are four cases. In the three first cases, we always have crossing Kenyon matchings as long as there is at least 1 occurrence of each letter in $pref_i(\bword)$. In the fourth case, it is only true for sufficiently large $i$:
\begin{figure}[b]
\includegraphics[width=\textwidth]{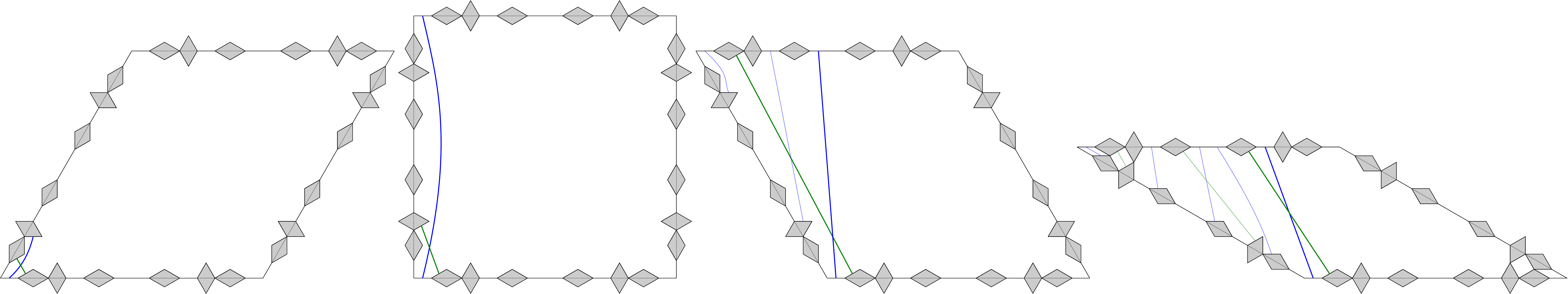}
\caption{For $n=6$ and the candidate edgeword $\prWord{6}$, in the corner of angle $\tfrac{k\pi}{6}$ for $k\in\{2,3,4,5\}$ the Kenyon matching of the first edge and the first rhombus cross.}
\label{fig:corner_condition_cases}
\end{figure}
\begin{enumerate}
\item $1\leq k < \tfrac{n}{2}$, see for example the corner of angle $\tfrac{2\pi}{6}$ in Figure \ref{fig:corner_condition_cases}. In this case, both the edge 0 and the rhombus 2 match to the adjacent side of the metatile. The two chains cross because the first edge 0 is matched to the first rhombus of type $2k$ on the adjacent edge, and the first rhombus $2$ is matched to the first rhombus of type $2k-2$. With $2k > 2k-2$, we know that the first rhombus of type $2k-2$ on the adjacent edge is before the first rhombus of type $2k$ in the edgeword. Which means that the Kenyon matchings of the first edge 0 and first rhombus 2 cross.
\item $k = \tfrac{n}{2}$, see for example the corner of angle $\tfrac{3\pi}{6}$ in Figure \ref{fig:corner_condition_cases}. In this  case, the edge 0 matches to the opposite side of the metatile and the rhombus 2 matches to the adjacent side of the metatile. So the Kenyon matchings cross.
\item $k = \tfrac{n}{2}+1$, see for example the corner of angle $\tfrac{4\pi}{6}$ in Figure \ref{fig:corner_condition_cases}. In this case, both the edge 0 and the rhombus 2 match to the opposite side of the metatile. But there are edges parallel to 0 on the adjacent side of the metatile: the rhombuses of type $2n-2$ on the adjacent side of the metatile have edges parallel to 0. However there is no edge parallel to the edge of rhombus 2 on the adjacent side of the metatile, because the edge of rhombus two is orthogonal to the adjacent side of the metatile. This means that the rhombus 2 matches to the first rhombus 2 on the opposite side of the metatile, but the edge 0 does not match to the first edge 0 of the opposite side of the metatile. Since the first rhombus 2 on the opposite side of the metatile is before the second 0 the Kenyon matchings cross.
\item $\tfrac{n}{2}+1 < k < n$, see for example the corner of angle $\tfrac{5\pi}{6}$ in Figure \ref{fig:corner_condition_cases}. In this case, both the edge 0 and the rhombus 2 match to the opposite side of the metatile and they both have parallel edges on the adjacent side of the metatile. On the adjacent side the edge direction 0 appears on rhombuses $2k'$ with $k':= n-k$, and the edge direction of rhombus 2 appears on rhombuses of type $2(k'+1)$. Let us now look at the Kenyon matchings of the first edge 0 and the first rhombus 2. We now use the counting functions of edgeword $\prWord{i}$, on the adjacent side there are $f_{2k'}(2i)$ rhombuses of type $2k'$ and $f_{2k'+2}(2i)$  rhombuses of type $2k'+2$ because the length of $\prWord{i}$ is $2i$. So the first edge 0 matches to $f_0^{-1}(f_{2k'}(2i)+1)$ and the first rhombus 2 matches to $f_2^{-1}(f_{2k'+2}(2i) + 1)$. For the two chains to cross we need
\[f_0^{-1}(f_{2k'}(2i)+1) > f_2^{-1}(f_{2k'+2}(2i) + 1).\]

Recall that the word $\prWord{i}$ approximates the billiard word $\bword$ and the counting functions $f_i(x)$ approximate $x\cdot \tgamma{i}$. More precisely there exists a bound $\delta>0$, only dependent on $k$ and $n$, such that
\begin{align*}
|f_0^{-1}(f_{2k'}(2i)+1)  - 2i \frac{\tgamma{2k'}}{\tgamma{0}} | &< \delta \\
|f_2^{-1}(f_{2k'+2}(2i)+1)  - 2i \frac{\tgamma{2k'+2}}{\tgamma{2}} | &< \delta
\end{align*}
Moreover, with $0<k' < \tfrac{n}{2}-1$, we have
\[ \frac{\tgamma{2k'}}{\tgamma{0}} - \frac{\tgamma{2k'+2}}{\tgamma{2}} = \cos\left(\frac{k'\pi}{n}\right) - \frac{\cos\left(\frac{(k'+1)\pi}{n}\right)}{\cos\left(\frac{\pi}{n}\right)} > 0\]
So for all sufficiently integer large $i$, we have $f_0^{-1}(f_{2k'}(2i)+1) > f_2^{-1}(f_{2k'+2}(2i) + 1)$, \ie the Kenyon matching of the first edge 0 and of the first rhombus 2 cross.
\end{enumerate}
At fixed $n$, there are only finitely many $k\in \{1,\dots, n-1\}$, so we can intersect the "for all sufficiently large integer $i$".
Overall, for all sufficiently large integer $i$ the edgeword $\prWord{i}$ induces a well-defined substitution that satisfies the corner condition.
\end{proof}

\subsection{Planar Rosa expansions and planarity}
\label{subsec:prPlanarity}
In this section we present the planarity of candidate Planar Rosa substitutions, \ie the fact that the candidate Planar Rosa substitution generate discrete plane tilings. The first step towards planarity is the lifting of the tilings and substitutions.

As we consider substitutions $\prSubs{}$ that are defined by an edgeword $u$, we can consider the induced expansion $\prExp{}{}$, we can lift the tilings, substitutions and expansions as described in Sections \ref{subsec:lifted_substitutions} and \ref{subsec:srLifting} and we can apply the propositions therein to obtain the following results:
\begin{enumerate}
\item the expansion matrix $\prMat{}$ is pseudo-circulant (Proposition \ref{prop:expansion_subrosa_even}),
\item it can be decomposed in a linear combination of the elementary matrices of Definition \ref{def:elementarymatrices} (Proposition \ref{prop:expansiondecomposition}),
\item with $\lambda = \tuple{\lambda_i}_{0\leq i < \nhalf}$ the vector of eigenvalues of $\prMat{}$ we have
\[ \eigenMatrix{n}\cdot [u]\trans = \lambda\trans,\]
with $\eigenMatrix{n}$ the matrix of Definition \ref{def:eigenvalue_matrix} 
(Proposition \ref{prop:eigenvalue}),
\item if $|\lambda_0| > 1$ and $|\lambda_i| < 1$ for all $1\leq i < \tfrac{n}{2}$, then the expansion $\prExp{}{}$ and the substitution $\prSubs{}$  are planar of slope $\slope_n^0$  (Proposition \ref{prop:planarity_non-planarity}).
\end{enumerate}

\begin{proposition}[Planarity]
\label{prop:prPlanarity}
There exist infinitely many $i\in\mathbb{N}$ such that the candidate edgeword $\prWord{i}$ defines a planar expansion of slope $\slope_n^0$, \ie assuming that the candidate substitution $\prSubs{i,n}$ exists, it generates discrete plane tilings of slope $\slope_n^0$.
\end{proposition}
Before proving this proposition, let us first present two intermediate results on the optimal frequency vector $\gamma$ and the candidate edgewords $\prWord{i}$.

\begin{lemma}[$\gamma$ is indeed the optimal frequency vector]
\label{lemma:gammaopt}
There exists $\epsilon>0$ and $l\in\mathbb{N}$ such that,
for any substitution $\sigma$ defined by an edgeword $u$ of length at least $l$,
we have
\[ d([u],\spanning{\gamma})<\epsilon \Rightarrow \sigma \text{ is planar of slope } \slope_n^0.\]
\end{lemma}
\begin{proof}[Proof of Lemma \ref{lemma:gammaopt}]
Let us first remark that \[ \eigenMatrix{n}\cdot \gamma\trans = \tuple{\nhalf, 0, 0,\dots 0}\trans .\]
This is due to the fact that $\eigenMatrix{n}$ is almost a Discrete Cosine Transform matrix and is almost orthogonal, see Lemma \ref{lemma:eigenMatrix_orthogonal} in the Appendices for more details.

By uniform continuity of the matrix-vector product, there exist $\epsilon>0$ such that for any two vector $x$ and $y$, $\|x-y\|<\epsilon \Rightarrow \|\eigenMatrix{n}\cdot x\trans - \eigenMatrix{n}\cdot y\trans\| < 1$.

Due to the facts that $\gamma$ has all positive coefficients and that $|u|$ is the 1-norm (sum of the absolute value of the coefficients) of $[u]$, there exists a length $l\in\mathbb{N}$ such that for any word $u$ of length at least $l$, if $d([u],\spanning{\gamma})<\epsilon$ then there exists $t> \tfrac{4}{n}$ such that $d([u],t\gamma)<\epsilon$.

Now take a substitution $\sigma$ of edgeword $u$ of length at least $l$ such that $d([u],\spanning{\gamma})<\epsilon$.
By Proposition \ref{prop:eigenvalue}, we know that $\sigma$ admits eigenspaces $\slope_n^k$ for $0\leq k < \tfrac{n}{2}$ with eigenvalues $\lambda = (\lambda_0,\cdots \lambda_{\nhalf -1}) = \eigenMatrix{n}\cdot [u]\trans $.

Since $|u|\geq l$, there exists $t> \tfrac{4}{n}$ such that $d([u],t\gamma)<\epsilon$,
so overall we have \\
$\| \lambda - t\eigenMatrix{n}\cdot \gamma\trans \| < 1$, \ie
\[ \left\| \lambda - \left(t\nhalf, 0,0, \dots 0\right)\right\| < 1,\]
in particular we get $\lambda_i < 1$ for any $1\leq i <\tfrac{n}{2}$, and with $t\tfrac{n}{2}> 2$ we get $\lambda_0 > 1$.
We now apply Proposition \ref{prop:planarity_non-planarity} to get that $\sigma$ is planar of slope $\slope_n^0$.

Note that the condition $|u|\geq l$ is only necessary for $\lambda_0>1$, \ie to ensure that the substitution has a scaling factor greater than 1 along the tiling plane $\slope_n^0$.
\end{proof}

\begin{lemma}[ $\abel{\prWord{i}}$ approximates $\spanning{\gamma}$]
\label{lemma:approximates}
For any $\epsilon > 0$, there exist infinitely many $i\in\mathbb{N}$ such that $d([\prWord{i}],\spanning{\gamma}) < \epsilon$.
\end{lemma}
\begin{proof}[Proof of Lemma \ref{lemma:approximates}]
To prove this result we define the billiard line or billiard sequence $(p_i)_{i\in\mathbb{N}}$ associated to the billiard word $\bword$ with starting point $0$ as the sequence of points of $\mathbb{Z}^n$ such that:
\begin{itemize}
\item $p_0 = 0$,
\item for any $i$, $\bword_i = 2k \Rightarrow p_{i+1} = p_{i} + \ee{k}$.
\end{itemize}

By definition of $\prWord{i}$ we have $[\prWord{i}] = 2p_{i}$. So $([\prWord{i}])_{i\in\mathbb{N}}$ approximates $\spanning{\gamma}$ arbitrarily well if and only if $(p_{i})_{i\in\mathbb{N}}$ approximates it arbitrarily well.
To prove that we use a known result on billiard words:
\begin{theorem}[folk.]
  Let $(p_i)_{i\in\mathbb{N}}$ be a billiard sequence of line $\Gamma$ and $\pi$ be the orthogonal projection onto $\Gamma^\bot$.
  Every projected point $\pi(p_i)$ is an accumulation point of the projected sequence $\left( \pi(p_i) \right)_{i \in \mathbb{N}}$.
  \label{th:folk}
\end{theorem}
Though this result is considered known by some, a quick proof using Poincaré's Recurrence Theorem can be found in \cite{kari2021} where it was also used in the same setting.

By Theorem \ref{th:folk}, $0=\pi(p_0)$ is an accumulation point of the projected sequence $\left(\pi(p_i)\right)_{i\in\mathbb{N}}$, so there exists a subsequence $\left(p_{\alpha(i)}\right)_{i\in\mathbb{N}}$ such that $\pi(p_{\alpha(i)}) \underset{i\to\infty}{\longrightarrow} 0$ which implies that $d(p_{\alpha(i)}, \spanning{\gamma}) \underset{i\to\infty}{\longrightarrow} 0$, \ie the sequence $\left(p_i\right)_{i\in\mathbb{N}}$ approximates $\spanning{\gamma}$ arbitrarily well.
\end{proof}
We can now combine these two lemmas to prove Proposition \ref{prop:prPlanarity}.
\begin{proof}[Proof of Proposition \ref{prop:prPlanarity}]
From Lemma \ref{lemma:gammaopt}, there exists $\epsilon>0$ such that for any susbtitution $\sigma$ defined by an edgeword $u$, if $d(\abel{u}, \langle\gamma\rangle) < \epsilon$ then $\sigma$ is planar of slope $\slope_n^0$.

From Lemma \ref{lemma:approximates}, and with $\epsilon$ provided by Lemma \ref{lemma:gammaopt}, there exist infinitely many $i\in\mathbb{N}$ such that $d(\abel{\prWord{i}}, \spanning{\gamma}) < \epsilon$.

Therefore there exist infinitely many $i\in\mathbb{N}$ such that, if the candidate substitution $\prSubs{i,n}$ of edgeword $\prWord{i}$ exists (if the metatiles are tileable), then it is planar of slope $\slope_n^0$.
\end{proof}


\subsection{Primitivity of the Planar Rosa substitution}
\label{subsec:prPrimitivity}

In this section we prove the primitivity of the Planar Rosa candidate substitutions. Recall that a substitution  $\sigma$ is called primitive when there exists an integer $k$ such that for any tile $t$, the metatile $\sigma^k(t)$ of order $k$ contains all the tiles in all possible orientations.

\begin{proposition}[Primitivity]
\label{prop:prPrimitivity}

There exists an integer $K'$ such that for all $i>K'$, if the candidate substitution $\prSubs{i,n}$ exists then it is primitive.
\end{proposition}

\begin{proof}
We prove here that for all sufficiently large $i$, the edgeword $\prWord{i}$ induces a substitution $\prSubs{i,n}$ that is primitive of order 2, \ie for any tile $t$ the patch ${\prSubs{i,n}}^2(t)$ contains all tiles in all orientations. Note that with similar strategy (and longer proof) we could also prove order 1 primitivity.

Let $i$ be an integer, we assume that the candidate substitution $\prSubs{i,n}$ exists and satisfies the corner condition, which by Proposition \ref{prop:prTileability} is true for all sufficiently large $i$.
Since $\prSubs{i,n}$ satisfies the corner condition we know that the image of any tile contains at least one occurrence of the narrow rhombus (in at least one orientation).
We now prove that for all sufficiently large $i$, the narrow metatile contains all tiles (in all orientations), which then implies that $\prSubs{i,n}$ is primitive of order 2.

Remark that the infinite cone of angle $\tfrac{\pi}{n}$ and edgeword $\bword$ in $\dualtiling{n}{\tfrac{1}{2}}$ contains every tile (in every orientation). This can be proved for example by uniform recurrence of the multigrid dual tiling $\dualtiling{n}{\tfrac{1}{2}}$.

Now take an oriented prototile $\mathbf{t}$ in the infinite cone. Let $t$ be the occurrence of $\mathbf{t}$ in the infinite cone closest to the corner. The tile
$t$ is the dual of an intersection of two lines in the  multigrid $\multigrid{n}{\tfrac{1}{2}}$. For all sufficiently large $i$, the same type of intersection appears in the Kenyon matching of the narrow metatile of $\prSubs{i,n}$ which means that there exist a tile $t' \equiv \mathbf{t}$ in the narrow metatile of $\prSubs{i,n}$.

\begin{figure}[t]
\center
\includegraphics[width=0.3\textwidth]{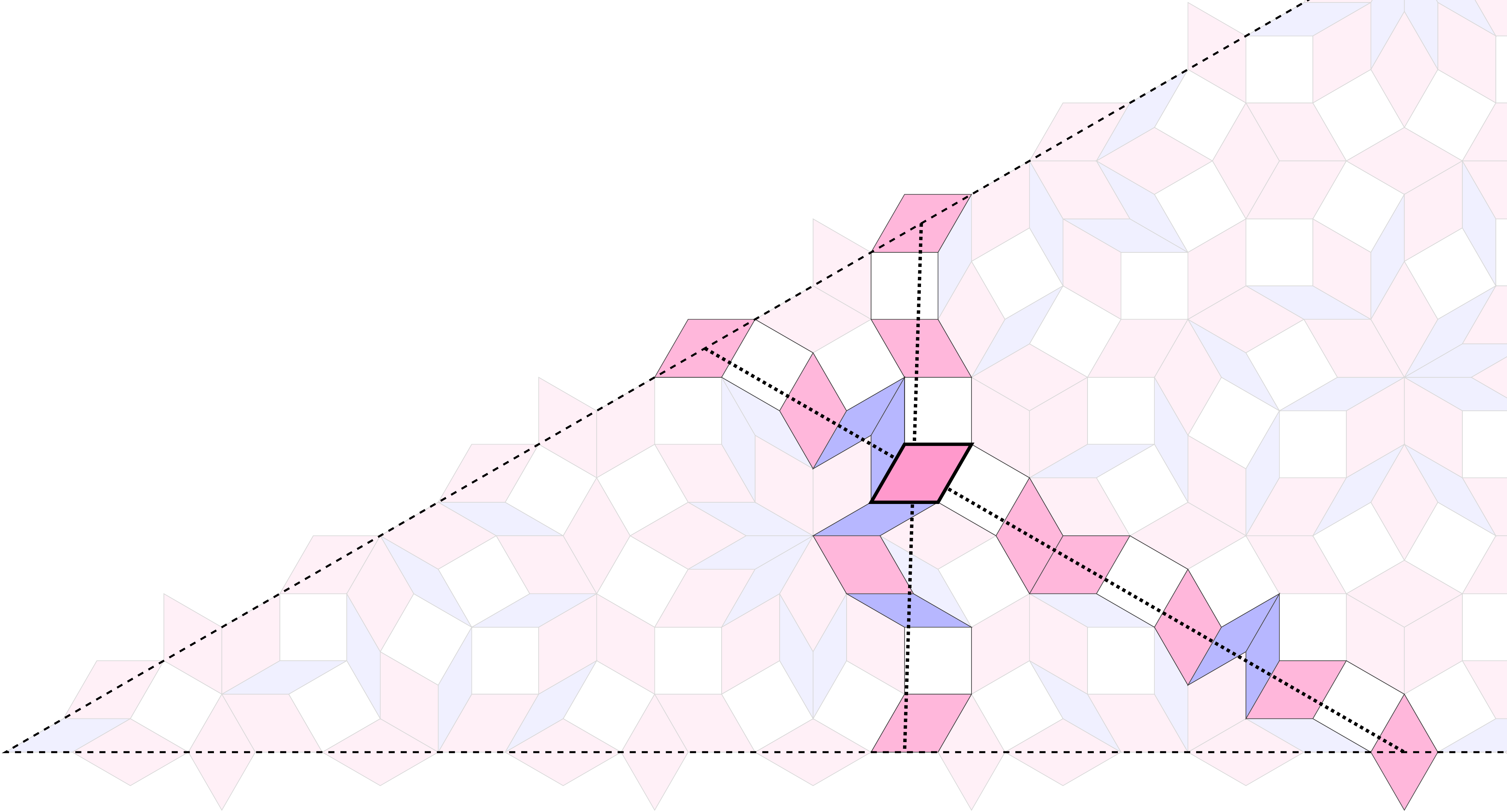} \hfill
\includegraphics[width=0.3\textwidth]{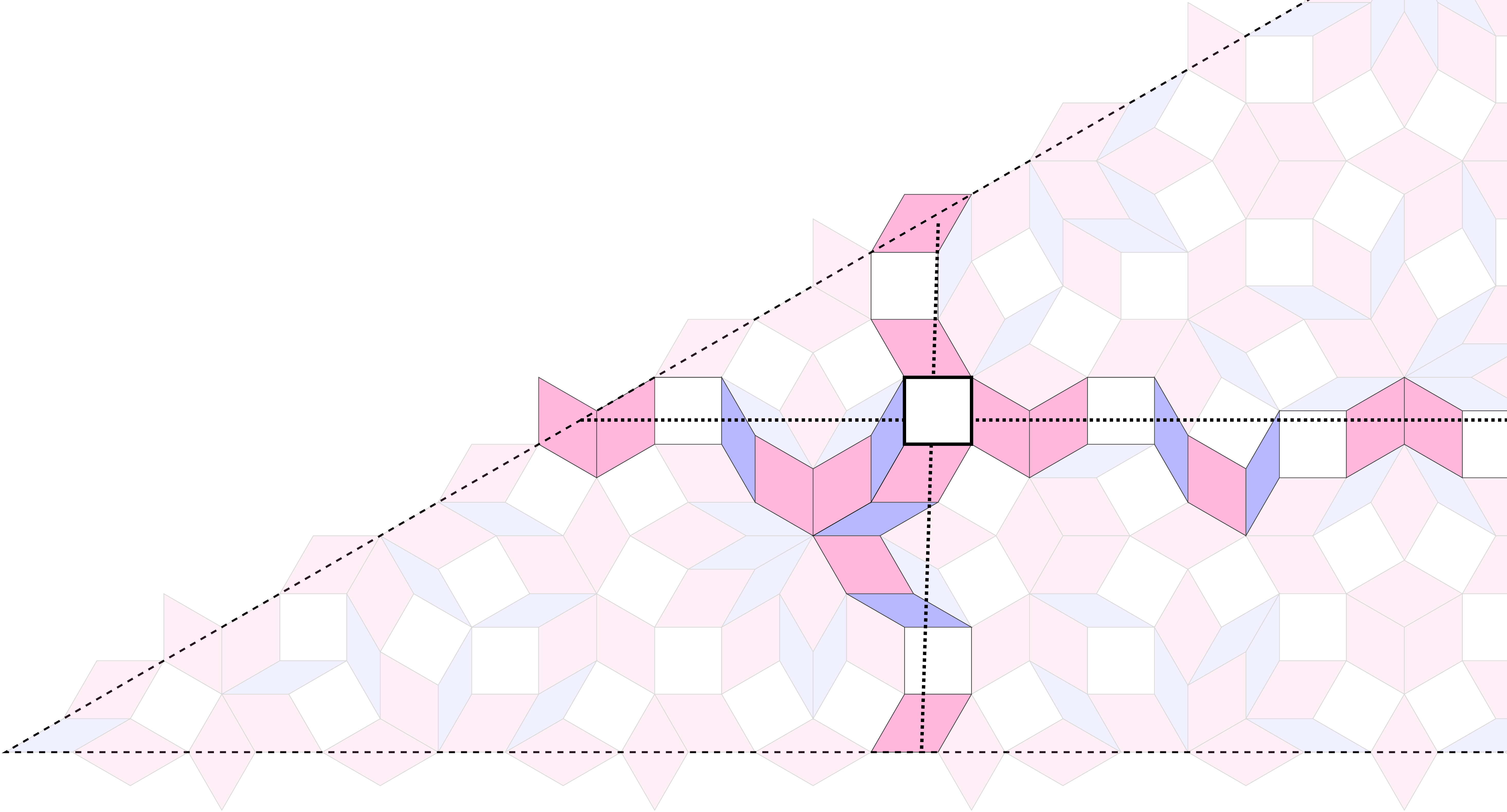} \hfill
\includegraphics[width=0.3\textwidth]{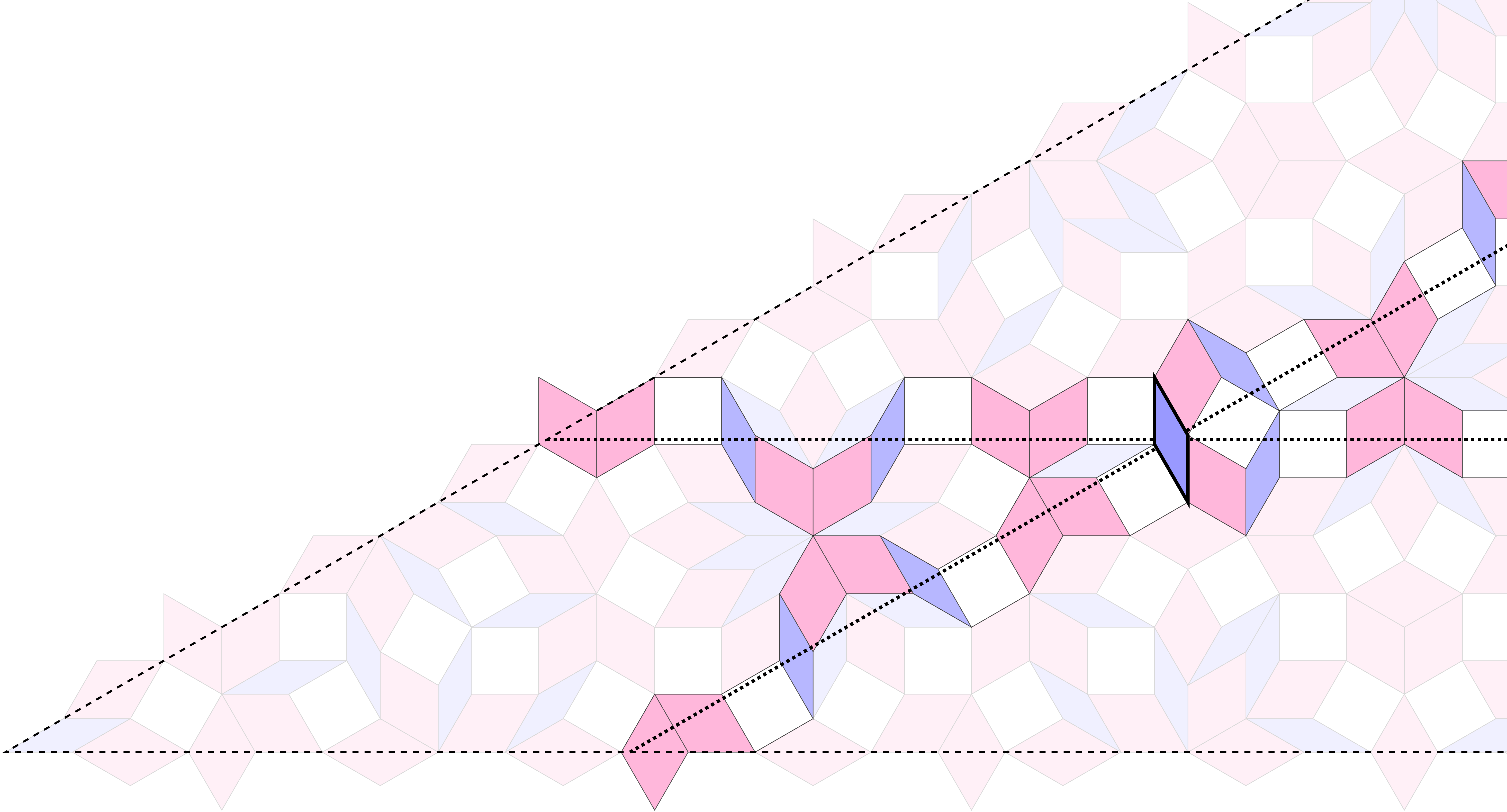}
\caption{Rhombuses in the infinite cone and chain intersections for $n=6$.}
\label{fig:chains_infinite_cone}
\end{figure}

In the cone, $t$ is the dual of an intersection of two lines, \ie at the intersection of two chains of rhombuses. For both chains of rhombuses let us look at the intersections of the chain with the boundary of the cone, there are two cases: either the chain crosses the cone, or the chain enters the cone and never leaves it, see Figure \ref{fig:chains_infinite_cone}. In the second case this mean that in the multigrid the line is parallel to one side of the cone. Indeed, because the cone has angle $\tfrac{\pi}{n}$ there is no line in the multigrid that is both non-parallel to the sides of the cone and has only one intersection point with the cone.

Denote by $x_t$ the maximal position (in the edgeword $\bword$) of the intersection points of the chains that cross at $t$ and the cone. For all $i> x_t$, if the narrow metatile induced by $\prWord{i}$ is tileable then it contains a tile $t'\equiv \mathbf{t}$. Indeed if the two chains cross the cone at positions $x_0,x_0'$ and $x_1,x_1'$ then the Kenyon matching of the narrow metatiles has matches $x_0$ to $x_0'$ and $x_1$ to $x_1'$ because $i> x_t$ and $\prWord{i}=pref_i(\bword)\overline{pref_i(\bword)}$, so we have an equivalent crossing, which means that the oriented prototile $\mathbf{t}$ appears. If one chain crosses the cone at positions $x_0,x_0'$ and one enters the cone at position $x_1$ and never leaves it, then we have an equivalent crossing because $x_1$ is matched to an edge or rhombus on the opposite side of the metatile so the two chains still cross. And if both chains enter at positions $x_0$ and $x_1$ and never leave the cone then they are both matched to edges or rhombuses on the opposite side of the metatile and the two chains still cross.
\end{proof}


\subsection{A seed for the Planar Rosa substitution}
\label{subsec:prSeed}

In this section we present the fact that the Star pattern is a regular seed for the Planar Rosa candidate substitutions.
Recall that the Star pattern consists of $n$ rhombuses of angle $\tfrac{\pi}{n}$ around a vertex, see Definition \ref{def:star} 
and Figure \ref{fig:star}.

\begin{proposition}[Seed]
\label{prop:prSeed}
Let $n$ be an even integer greater than 2.\\
For any $i$, if $\prSubs{i,n}$ is a well-defined substitution that satisfies the corner condition, then the star pattern $\star{n}$ is a regular seed for $\prSubs{i,n}$.
\end{proposition}

\begin{proof}
Let $n\geq 4$ be an even integer and let $i$ be an integer such that the candidate substitution $\prSubs{i,n}$ exists and satisfies the corner condition.

We prove the two conditions:
\begin{enumerate}
\item $\star{n}$ is a seed for $\prSubs{i,n}$, \ie $\star{n}$ appears at the centre of $\prSubs{i,n}(\star{n})$ and $\prSubs{i,n}$ expands $\star{n}$ in all directions,
\item $\star{n}$ is regular for $\prSubs{i,n}$, \ie there exists a prototile $t$ and an integer $m$ such that $\star{n}$ appears in ${\prSubs{i,n}}^m(t)$.
\end{enumerate}

The fact that $\star{n}$ is a seed is an immediate corollary of the corner condition. Indeed in the narrow corner of the narrow metatile we have a narrow rhombus, so when we apply $\prSubs{i,n}$ to $\star{n}$ we have $\star{n}$ at the centre. Moreover since the expansion of $\prSubs{i,n}$ is a linear application that has scaling factor strictly greater than 2 the substitution expands $\star{n}$ in all direction.

\begin{figure}[t]
\center
\includegraphics[width=0.6\textwidth]{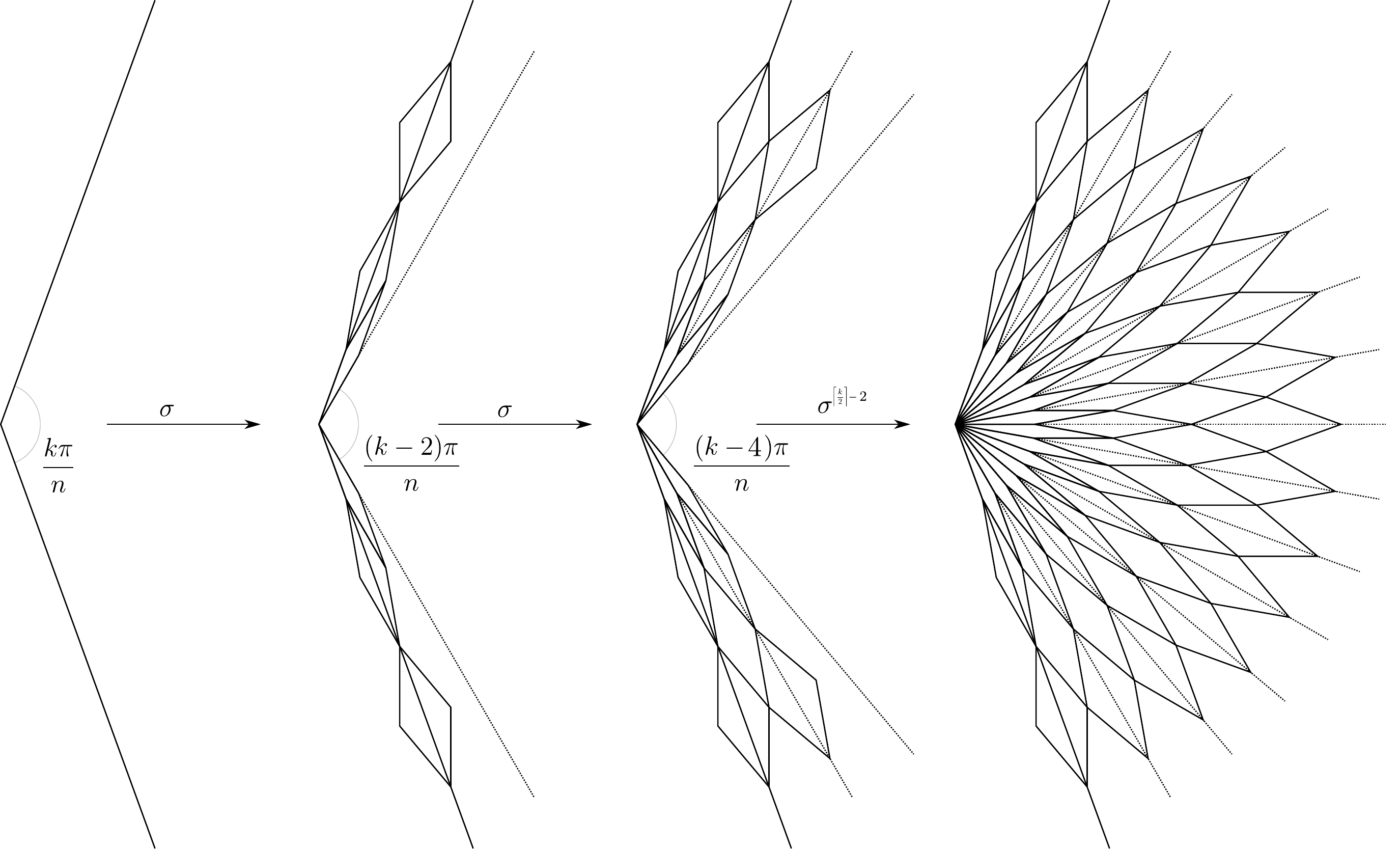}
\caption{A corner of metatile when applying a substitution that satisfies the corner condition.}
\label{fig:corners_metatile}
\end{figure}

The fact that $\star{n}$ is regular is a bit more tricky but is also a corollary of the corner condition. We prove that for any prototile $t$, $\star{n}$ appears in the $(\tfrac{n}{2}+1)$'th image of $t$.

First, consider a tile $t'$ of angles $\tfrac{k\pi}{n}$ and $\tfrac{(n-k)\pi}{n}$, in the $\left\lfloor \tfrac{k}{2}\right\rfloor$'th image of $t'$ by the substitution $\prSubs{i,n}$ in the two opposite corners of angle $\tfrac{k\pi}{n}$ we have a $\tfrac{k\pi}{n}$ portions of the star pattern $\star{n}$.
Indeed by the corner condition, in the corner of angle $\tfrac{k\pi}{n}$ of the metatile $\prSubs{i,n}(t')$ there is a narrow $\tfrac{\pi}{n}$ rhombus of both sides and a remaining undetermined portion of angle $\tfrac{(k-2)\pi}{n}$ which can be composed of one or more tiles. Now in ${\prSubs{i,n}}^2(t')$ with the same arguments in the $\tfrac{k\pi}{n}$ corner there are at least two narrow $\tfrac{\pi}{n}$ rhombuses on each side and a remaining undertermined portion of angle $\tfrac{(k-4)\pi}{n}$ which can be composed of one or more tiles. By repeating this process we obtain that in the $\tfrac{k\pi}{n}$ corner of the $\left\lfloor \tfrac{k}{2}\right\rfloor$'th image of $t'$ by the substitution $\prSubs{i,n}$ we have $k$ narrow $\tfrac{\pi}{n}$ rhombuses, \ie a $\tfrac{k\pi}{n}$ portion of the star pattern $\star{n}$. See Figure \ref{fig:corners_metatile} for an illustration of this construction.

Now take a prototile $t$, in $\prSubs{i,n}(t)$ there is at least one interior vertex. This interior vertex is surrounded by at least 3 tiles, for simplicity we assume that it is surrounded by exactly three tiles of angles $\tfrac{k_0\pi}{n}$, $\tfrac{k_1\pi}{n}$ and $\tfrac{k_2\pi}{n}$ with $1\leq k_0,k_1,k_2 \leq n-1$,  so we apply the previous construction and in the $(\tfrac{n}{2}+1)$'th image of $t$ we have a full star $\star{n}$ centred on the $\tfrac{n}{2}$'th image of that interior vertex.
\end{proof}


\subsection{The Planar Rosa substitution and canonical tiling}
\label{subsec:prConclusion}
In this section we combine the results of Sections \ref{subsec:prTileability} to \ref{subsec:prSeed} to prove the correctness of the Planar Rosa construction.

Recall that the Planar Rosa substitution $\prSubs{n}$ is defined as the first candidate substitution $\prSubs{i,n}$ (defined by the candidate edgeword $\prWord{i}$) such that it is a a primitive substitution that satisfies the corner condition and it is planar substitution of slope $\slope_n^0$.

\begin{proposition}[The construction of Planar Rosa is correct]
The Planar Rosa substitution $\prSubs{n}$ exists.
\end{proposition}
\begin{proof}
We combine the previous results:
\begin{itemize}
\item from Proposition \ref{prop:prTileability} there exists an integer $K$ such that the set of integers $I_1:= [K, +\infty[$ satisfies $\forall i \in I_1,$ the candidate substitution $\prSubs{i,n}$ exists, \ie the metatiles induced by the candidate edgeword $\prWord{i}$ are tileable in a way that satisfies the corner condition,
\item from Proposition \ref{prop:prTileability} there exists an infinite set of integers $I_2$ such that $\forall i\in I_2$, if the candidate substitution $\prSubs{i,n}$ exists then it is planar of slope $\slope_n^k$,
\item from Proposition \ref{prop:prPrimitivity} there exists an integer $K'$ such that the sef of integers $I_3:= [K',+\infty[$ satisfies $\forall i \in I_1,$ if the candidate substitution $\prSubs{i,n}$ exists then it is primitive.
\end{itemize}
Remark that $I_1\cap I_2 \cap I_3$ is non-empty because any infinite subset of $\mathbb{N}$ has a non-empty intersection with any infinite interval $[k,+\infty[$.

Take $i \in I_1 \cap I_2 \cap I_3$, by construction the candidate substitution $\prSubs{i,n}$ has all the desired properties.
Since $I_1 \cap I_2 \cap I_3$ is a non-empty subset of $\mathbb{N}$ it admits a minimum. So the Planar Rosa substitution $\prSubs{n}$ exists.
\end{proof}
Recall that by Proposition \ref{prop:prSeed} the star pattern $\star{n}$ is a legal seed for $\prSubs{n}$, so the canonical Planar Rosa tiling $\prTiling{n}$ defined as the fixpoint of $\prSubs{n}$ from the seed $\star{n}$ is a substitution discrete plane of slope $\slope_n^0$ with global $2n$-fold rotational symmetry. This concludes the proof of Theorem \ref{th:planar-rosa}.


\appendix
\section{Technical details}

\begin{lemma}[Decompposition of $\mathbb{R}^n$ as the orthogonal direct sum of the spaces $\slope_n^k$]
\label{lemma:decomposition_Rn}
For any even integer $n$.
Define the planes $\slope_n^k$ for $0\leq k < \tfrac{n}{2}$ as
\[ \slope_n^k := \left\langle \left(\cos\frac{(2k+1)i\pi}{n}\right)_{0\leq i < n}, \left(\sin\frac{(2k+1)i\pi}{n}\right)_{0\leq i < n} \right\rangle. \]
The planes $\slope_n^k$ are pairwise orthogonal and their direct sum is $\mathbb{R}^n$.
\end{lemma}
\begin{proof}
Remark that the fact that the direct sum of the spaces $\slope_n^k$ is $\mathbb{R}^n$ is a consequence of the fact that they are pairwise orthogonal. Indeed if they are pairwise orthogonal then their sum is direct (trivial intersection) and the dimension of the sum is the sum of the dimensions so the dimension is $n$ and the sum is $\mathbb{R}^n$.

Now we prove that they are pairwise orthogonal. Take $0\leq k_1 < k_2 < \tfrac{n}{2}$.
$\slope_n^{k_1}$ and $\slope_n^{k_2}$ are orthogonal, indeed :
\begin{itemize}
\item $\left\langle \left(\cos\frac{(2k_1+1)i\pi}{n}\right)_{0\leq i < n} \bigg|  \left(\cos\frac{(2k_2+1)i\pi}{n}\right)_{0\leq i < n} \right\rangle = 0$, indeed we have
\begin{align*}
&\left\langle \left(\cos\frac{(2k_1+1)i\pi}{n}\right)_{0\leq i < n} \bigg|  \left(\cos\frac{(2k_2+1)i\pi}{n}\right)_{0\leq i < n} \right\rangle \\
&\qquad = \sum\limits_{i=0}^{n-1} \cos\frac{(2k_1+1)i\pi}{n}\cos\frac{(2k_2+1)i\pi}{n}\\
&\qquad = \sum\limits_{i=0}^{n-1}\frac{1}{2}\left( \cos\frac{2i(k_1+k_2+1)\pi}{n}+ \cos\frac{2i(k_2-k_1)\pi}{n}\right)\\
&\qquad = \left(\sum\limits_{i=0}^{n-1}\frac{1}{2} \cos\frac{2i(k_1+k_2+1)\pi}{n}\right)+ \left(  \sum\limits_{i=0}^{n-1}\frac{1}{2}\cos\frac{2i(k_2-k_1)\pi}{n}\right)\\
&\qquad = 0
\end{align*}
Note that the important hypothesis is $0\leq k_1 < k_2 < \tfrac{n}{2}$ so $0< (k_2 - k_1), (k_1+k_2+1) < n$ and $\sum\limits_{i=0}^{n-1}\frac{1}{2} \cos\frac{2i(k_1+k_2+1)\pi}{n} = 0$ (and same with $(k_2-k_1$)).
\item similarly $\left\langle \left(\sin\frac{(2k_1+1)i\pi}{n}\right)_{0\leq i < n} \bigg|  \left(\sin\frac{(2k_2+1)i\pi}{n}\right)_{0\leq i < n} \right\rangle = 0$ with \[2\sin\frac{(2k_1+1)i\pi}{n}\sin\frac{(2k_2+1)i\pi}{n} = \cos\frac{2i(k_1+k_2+1)\pi}{n}- \cos\frac{2i(k_2-k_1)\pi}{n}.\]
\item similarly we have $\left\langle \left(\cos\frac{(2k_1+1)i\pi}{n}\right)_{0\leq i < n} \bigg|  \left(\sin\frac{(2k_2+1)i\pi}{n}\right)_{0\leq i < n} \right\rangle = 0$ with
\[2\cos\frac{(2k_1+1)i\pi}{n}\sin\frac{(2k_2+1)i\pi}{n} = \sin\frac{2i(k_1+k_2+1)\pi}{n}- \sin\frac{2i(k_2-k_1)\pi}{n}.\]
\item similarly we have $\left\langle \left(\sin\frac{(2k_1+1)i\pi}{n}\right)_{0\leq i < n} \bigg|  \left(\cos\frac{(2k_2+1)i\pi}{n}\right)_{0\leq i < n} \right\rangle = 0$.
\end{itemize}

\end{proof}

\begin{lemma}[The billiard word $\bword$ of slope $\gamma$ is well-defined ]
Let $n$ be an even integer at least 4.
Let $\gamma:=(\cos\tfrac{i\pi}{n})_{0\leq i < n/2}$.
Let $\Gamma_\half := \langle \gamma \rangle + (\half,\half, \dots ,\half) = \{ (\half + t\cos\tfrac{i\pi}{n})_{0\leq i < n/2} |\ t \in \mathbb{R}^+\}$.
Let $\hyperplane{j}{k}:= \{ x \in \mathbb{R}^{n/2}|\ x_j = k\}$ with $j<\nhalf$ and $k> 1$.

The billiard word $\bword$ of line $\Gamma_\half$ is well-defined , \ie the line $\Gamma_\half$ never intersects two distinct hyperplanes at once, \ie \\ for any $j,j'\in \{0,1, \dots, \nhalf -1\}$, $k,k' \in \mathbb{N}$, if $(j,k)\neq (j',k')$ then $\Gamma_\half \cap \hyperplane{j}{k} \cap \hyperplane{j'}{k'} = \emptyset$.
\label{lemma:billiard_word_well_defined}
\end{lemma}
\begin{proof}
If $j=j'$ the result is trivial because in that case $k\neq k'$ and therefore $\hyperplane{j}{k} \cap \hyperplane{j'}{k'} = \emptyset$.

If $j\neq j'$, without loss of generality we assume $j<j'$. For contradiction, assume that the intersection is non-empty, \ie assume that there exists $x\in \Gamma_\half \cap \hyperplane{j}{k} \cap \hyperplane{j'}{k'}$.
By definition of $\Gamma_\half$, $x = t\gamma + (\half,\half,\dots,\half)$ for some $t \in \mathbb{R}^+$.
By definition of $\hyperplane{j}{k}$ and $\hyperplane{j'}{k'}$, $x_j = k$ and $x_{j'} = k'$.

So we get $t = (k-\half)/\cos(\tfrac{j\pi}{n}) = (k'-\half)/\cos(\tfrac{j'\pi}{n})$ which implies $(k-\half)\cos(\tfrac{j'\pi}{n}) - (k'-\half)\cos(\tfrac{j\pi}{n}) = 0$.
Since the angles are rational with $\pi$ with $0\leq \tfrac{j\pi}{n} < \tfrac{j'\pi}{n} < \tfrac{\pi}{2}$ and the coefficients are rational, this is a trigonometric diophantine equation with two terms so we can apply known results on trigonometric diophantine equations \cite{conway1976}.
Using this known result we obtain that the only possible case is $\tfrac{j\pi}{n}=0$ and $\tfrac{j'\pi}{n}=\tfrac{\pi}{3}$.
We now have $(k-\half)\cdot \half - (k'-\half) = 0$, \ie $\tfrac{k}{2} + \tfrac{1}{4} = k'$ which is impossible since $k$ and $k'$ are integers.
We get the expected contradiction so this means that the intersection is empty.
\end{proof}

\begin{lemma}[The eigenvalue matrix $\eigenMatrix{n}$ is almost orthogonal]
Let $n$ be an even integer greater than 2. Let $\eigenMatrix{n}$ be defined as
\[\eigenMatrix{n} := \left(\elementaryEigenvalue{n}{j}{i} \right)_{0\leq i,j < \frac{n}{2}} = \left(\eta_j \cos \tfrac{(2i+1)j\pi}{n}\right)_{0\leq i,j < \frac{n}{2}},\]
with $\eta_0:=1$ and $\eta_j:= 2$ for $0<j<\tfrac{n}{2}$ (Definition \ref{def:eigenvalue_matrix}).

The matrix $\eigenMatrix{n}$ is almost orthogonal in the sense that $\eigenMatrix{n}\cdot D\trans = D\cdot \eigenMatrix{n}\trans = \frac{n}{2}\mathrm{Id}_{\frac{n}{2}}$ with the matrix $D$ defined below.

In particular, $\eigenMatrix{n}\cdot \gamma\trans = (\frac{n}{2},0,0,\dots 0)\trans$ with $\gamma:= \tuple{\cos(\tfrac{i\pi}{n})}_{0\leq i < \frac{n}{2}}$ the optimal rhombus frequency vector (Definition \ref{def:gamma}). 
\label{lemma:eigenMatrix_orthogonal}
\end{lemma}

\begin{proof}
Let us first recall that here $n$ is an even integer.

\emph{Definitions:} We define four matrices $A$, $B$, $C$ and $D$ by:
\begin{align*}
A&:= \left(a_i \cos\frac{i(2j+1)\pi}{2n}\right)_{0\leq i,j < n} \qquad \qquad \qquad a_i = \begin{cases} \sqrt{\frac{1}{n}} \text{ if } i=0 \\ \sqrt{\frac{2}{n}} \text{ otherwise}\end{cases}\\
B&:= \left(b_i \cos\frac{i(2j+1)\pi}{n}\right)_{0\leq i< \frac{n}{2},0\leq j < n} \qquad \qquad b_i = \begin{cases} \sqrt{\frac{1}{n}} \text{ if } i=0 \\ \sqrt{\frac{2}{n}} \text{ otherwise}\end{cases}\\
C&:= \left(c_i \cos\frac{i(2j+1)\pi}{n}\right)_{0\leq i,j< \frac{n}{2}} \qquad \qquad \qquad c_i = \begin{cases} \sqrt{\frac{2}{n}} \text{ if } i=0 \\ \frac{2}{\sqrt{n}} \text{ otherwise}\end{cases}\\
D&:= \left(\cos\frac{(2i+1)j\pi}{n}\right)_{0\leq i,j< \frac{n}{2}}
\end{align*}

\emph{Overview of the proof:}  the matrix $A$ is a Discrete Cosine Transform matrix, sometimes called DCT-III, which is know to be orthogonal. From this we prove that $B$ is \emph{semi-orthogonal} and $C$ is orthogonal. From the fact that $C$ is orthogonal we get $\eigenMatrix{n}\cdot D\trans = D\cdot \eigenMatrix{n}\trans = \tfrac{n}{2}\mathrm{Id}_{\frac{n}{2}}$ because $(\eigenMatrix{n}\cdot D\trans)_{i,k} = \tfrac{n}{2}(C\cdot C\trans)_{i,k}$.

\emph{Details and computations:} We consider as known the fact that $A$ (DCT-III) is orthogonal.
As $B$ is a rectangular matrix and not a square matrix, it cannot be orthogonal, however we say it is semi-orthogonal for $B\cdot B\trans = \mathrm{Id}_{\frac{n}{2}}$.
This holds due to the fact that $A$ is orthogonal and
\[(B\cdot B\trans)_{i,k} = \sum\limits_{j=0}^{n-1}B_{i,j}B_{k,j} = \sum\limits_{j=0}^{n-1}A_{2i,j}A_{2k,j} = (A\cdot A\trans)_{2i,2k}.\]
Now remark that $B_{i,j}=B_{i,n-1-j}$ because
\[
\cos\left(\frac{i(2(n-1-j)+1)\pi}{n}\right)
= \cos\left(2i\pi - \frac{i(2j+1)\pi}{n}\right)
= \cos\left(\frac{i(2j+1)\pi}{n}\right),
\]

So $(B\cdot B\trans)_{i,k} = (C\cdot C\trans)_{i,k}$ because $C_{i,j}=\sqrt{2}B_{i,j}$ and
\[
(B\cdot B\trans)_{i,k}
= \sum\limits_{j=0}^{n-1} B_{i,j}B_{k,j}
= 2\sum\limits_{j=0}^{\frac{n}{2}-1} B_{i,j}B_{k,j}
= \sum\limits_{j=0}^{\frac{n}{2}-1} \sqrt{2}B_{i,j}\sqrt{2}B_{k,j}
= \sum\limits_{j=0}^{\frac{n}{2}-1} C_{i,j}C_{k,j}.
\]
So $C$ is orthogonal. Now remark that $C$ is somewhat similar to $\eigenMatrix{n}\trans$.\\
We actually have $(C\trans\cdot C)_{i,j} = \tfrac{2}{n} (\eigenMatrix{n}\cdot D\trans)_{i,j}$ , for any $0\leq i,j < n$, indeed
\begin{align*}
\left(C\trans \cdot C\right)_{i,j}
&= \sum\limits_{k=0}^{\frac{n}{2}-1} C_{k,i}C_{k,j}
= \sum\limits_{k=0}^{\frac{n}{2}-1} c_k^2 \cos\left(\frac{k(2i+1)\pi}{n}\right)\cos\left(\frac{k(2j+1)\pi}{n}\right)\\
&= \sum\limits_{k=0}^{\frac{n}{2}-1} \eta_k\frac{2}{n} \cos\left(\frac{k(2i+1)\pi}{n}\right)\cos\left(\frac{k(2j+1)\pi}{n}\right)
= \tfrac{2}{n} \left(\eigenMatrix{n}\cdot D\trans\right)_{i,j},
\end{align*}
recall that $\eigenMatrix{n}=(\eta_j \cos\tfrac{(2i+1)j\pi}{n})_{0\leq i,j < \frac{n}{2}}$ with $\eta_0 = 1$ and $\eta_j= 2$ for $j>0$, so we indeed have $c_j^2 = \tfrac{2 \eta_j}{n}$.

In particular we have $\eigenMatrix{n}\cdot \gamma\trans = (\tfrac{n}{2},0,0\dots, 0)\trans$.
\end{proof}

\bibliography{planar-rosa_article}

\begin{thebibliography}{10}

\bibitem{arnoux2011}
P.~Arnoux, M.~Furukado, E.~Harriss, and S.~Ito.
\newblock Algebraic numbers, free group automorphisms and substitutions on the
  plane.
\newblock {\em Transactions of the American Mathematical Society}, 2011.

\bibitem{arnoux2001higher}
P.~Arnoux, S.~Ito, and Y.~Sano.
\newblock Higher dimensional extensions of substitutions and their dual maps.
\newblock {\em Journal d’Analyse Math{\'e}matique}, 2001.

\bibitem{baake2013}
M.~Baake and U.~Grimm.
\newblock {\em Aperiodic Order: A Mathematical Invitation}.
\newblock Cambridge University Press, 2013.

\bibitem{beenker1982}
F.~P.~M. Beenker.
\newblock Algebraic theory of non-periodic tilings of the plane by two simple
  building blocks: a square and a rhombus.
\newblock {\em EUT report. WSK, Dept. of Mathematics and Computing Science},
  1982.

\bibitem{bedaride2015}
N.~Bédaride and Th. Fernique.
\newblock When periodicities enforce aperiodicity.
\newblock {\em Communications in Mathematical Physics}, 2015.

\bibitem{conway1976}
J.~Conway and A.~Jones.
\newblock Trigonometric diophantine equations (on vanishing sums of roots of
  unity).
\newblock {\em Acta Arithmetica}, 1976.

\bibitem{Davis1979}
Philip~J. Davis.
\newblock {\em Circulant Matrices}.
\newblock Wiley, New York, 1979.

\bibitem{debruijn1981}
N.~G. De~Bruijn.
\newblock {Algebraic theory of Penrose’s nonperiodic tilings of the plane. I
  and II.}
\newblock {\em Kon. Nederl. Akad. Wetensch. Proc. Ser. A}, 1981.

\bibitem{gahler1986}
F.~Gahler and J.~Rhyner.
\newblock Equivalence of the generalised grid and projection methods for the
  construction of quasiperiodic tilings.
\newblock {\em Journal of Physics A: Mathematical and General}, 1986.
\newblock
  \href{https://dx.doi.org/10.1088/0305-4470/19/2/020}{doi:10.1088/0305-4470/19/2/020}.

\bibitem{grunbaum1987}
B.~Gr{\"u}nbaum and G.~C. Shephard.
\newblock {\em Tilings and patterns}.
\newblock Courier Dover Publications, 1987.

\bibitem{kari2021}
J.~Kari and V.~H. Lutfalla.
\newblock Substitution discrete plane tilings with $2n$-fold rotational
  symmetry for odd $n$.
\newblock {\em Discrete \& Computational Geometry}, 2022.

\bibitem{kari2016}
J.~Kari and M.~Rissanen.
\newblock {Sub Rosa}, a system of quasiperiodic rhombic substitution tilings
  with $n$-fold rotational symmetry.
\newblock {\em Discrete {\&} Computational Geometry}, 2016.

\bibitem{kenyon1993}
R.~Kenyon.
\newblock Tiling a polygon with parallelograms.
\newblock {\em Algorithmica}, 1993.

\bibitem{levitov1988}
L.~S. Levitov.
\newblock Local rules for quasicrystals.
\newblock {\em Communications in mathematical physics}, 1988.

\bibitem{lutfalla2021multigrid}
V.~H. Lutfalla.
\newblock {An Effective Construction for Cut-And-Project Rhombus Tilings with
  Global $n$-Fold Rotational Symmetry}.
\newblock In {\em 27th IFIP WG 1.5 International Workshop on Cellular Automata
  and Discrete Complex Systems (AUTOMATA 2021)}.

\bibitem{lutfalla2021these}
V.~H. Lutfalla.
\newblock {\em {Substitution discrete planes}}.
\newblock Theses, {Universit{\'e} Sorbonne Paris Nord}, July 2021.
\newblock
  \href{https://hal.archives-ouvertes.fr/tel-03376430}{hal:tel-03376430}.

\bibitem{penrose1974}
R.~Penrose.
\newblock The role of aesthetics in pure and applied mathematical research.
\newblock {\em Bull. Inst. Math. Appl.}, 1974.

\bibitem{senechal1996}
M.~Senechal.
\newblock {\em Quasicrystals and geometry}.
\newblock CUP Archive, 1996.

\bibitem{socolar1990}
J.~E.~S. Socolar.
\newblock Weak matching rules for quasicrystals.
\newblock {\em Communications in mathematical physics}, 1990.

\bibitem{solomyak1998}
B.~Solomyak.
\newblock Nonperiodicity implies unique composition for self-similar
  translationally finite tilings.
\newblock {\em Discrete \& Computational Geometry}, 1998.

\end{thebibliography}

\end{document}